\documentclass[10pt,journal,compsoc]{IEEEtran}

\ifCLASSOPTIONcompsoc
  \usepackage[nocompress]{cite}
\else
  \usepackage{cite}
\fi

\ifCLASSINFOpdf
\else
\fi

\usepackage{wrapfig}
\usepackage{graphicx}
\usepackage{amsmath}
\usepackage{amssymb}
\usepackage{amsthm}
\usepackage{multirow}
\usepackage{lipsum}
\usepackage{xcolor,colortbl}
\usepackage{enumitem}

\usepackage{algorithmic}
\usepackage{color}
\usepackage{verbatim}

\newtheorem{lemma}{Lemma}

\newtheorem{remark}{Remark}

\DeclareMathOperator*{\argmax}{arg\,max}

\definecolor{Gray}{gray}{0.8}

\makeatletter
\def\maketag@@@#1{\hbox{\m@th\normalfont\normalsize#1}}
\makeatother

 \usepackage{framed}

\begin{document}

\title{Optimal Multiphase Investment Strategies for Influencing Opinions in a Social Network}
\author{Swapnil Dhamal,
Walid Ben-Ameur,
Tijani Chahed, and 
Eitan Altman 
\IEEEcompsocitemizethanks{\IEEEcompsocthanksitem S. Dhamal, W. Ben-Ameur, and T. Chahed are with Samovar, T´el´ecom SudParis, CNRS, Universit´e Paris-Saclay, France.
E. Altman is with Institut National de Recherche en Informatique et en Automatique (INRIA) Sophia
Antipolis-M´editerran´ee, France.
}
}

%

\IEEEtitleabstractindextext{%
\begin{abstract}
We study the problem of optimally investing in nodes of a social network in a competitive setting, where two camps aim to maximize adoption of their opinions by the population. In particular, we consider the possibility of campaigning in multiple phases, where the final opinion of a node in a phase acts as its initial biased opinion for the following phase. Using an extension of the popular DeGroot-Friedkin model, we formulate the utility functions of the camps, and show that they involve what can be interpreted as multiphase Katz centrality. Focusing on two phases, we analytically derive Nash equilibrium investment strategies, and the extent of loss that a camp would incur if it acted myopically. Our simulation study affirms that nodes attributing higher weightage to initial biases necessitate higher investment in the first phase, so as to influence these biases for the terminal phase. We then study the setting in which a camp's influence on a node depends on its initial bias. For single camp, we present a polynomial time algorithm for determining an optimal way to split the budget between the two phases. For competing camps, we show the existence of Nash equilibria under reasonable assumptions, and that they can be computed in polynomial time.
\end{abstract}

}

\maketitle

\IEEEdisplaynontitleabstractindextext

\IEEEpeerreviewmaketitle

\vspace{-1mm}
\IEEEraisesectionheading{\section{Introduction}}
\label{sec:ODSNmultiphase_intro}

\IEEEPARstart{T}{he} study of opinion dynamics in a society has been of prime importance 
to understand and influence the processes and outcomes of elections, viral marketing, propagation of ideas and behaviors, etc. 
In this paper, we consider two competing camps who aim to maximize the adoption of their respective opinions in a social network. In particular, we consider a strict competition setting where the opinion value of one camp is denoted by $+1$ and that of the other camp by $-1$; we refer to these camps as good and bad camps respectively.
%
%
%
%
Opinion adoption by a population can be quantified in a variety of ways based on the  application;
here we consider a well-accepted way, namely, 
the average or equivalently, the sum of the opinion values of the nodes in the network \cite{gionis2013opinion,grabisch2017strategic}.


The average or sum of the opinion values is of relevance in several applications, for example,  a fund collection scenario where the magnitude of a node's opinion value corresponds to the amount and the sign indicates the camp towards which it is willing to contribute.
Another example is when nodes are agents reporting the 
intensity of an
event based on the information they receive individually; this information could be influenced by the camps.
Other examples are elections and product purchases, where the bounded opinion value of a node would  translate into its probability of voting for a particular party or purchasing  a particular product, and so the sum of  opinion values would  translate into the expected number of votes or purchases in favor of a party or product.


Social networks play a prime role in determining the opinions of the constituent nodes, since nodes usually update their opinions based on the opinions of their connections \cite{easley2010networks,acemoglu2011opinion}. 
This fact is exploited by a camp who intends to influence the opinions of these nodes in its favor.
%
 %
 A camp could determine nodes whose opinion values it would want to directly influence (without the aid of social network) by investing on them in the form of money, free products or discounts, attention, convincing discussions, etc. 
 %
Thus given a budget constraint, the strategy of a camp comprises of how much to invest on a given node, in presence of a competing camp who also would invest strategically, so as to maximize its opinion adoption. 
As per the popular DeGroot-Friedkin model \cite{friedkin1990social,friedkin1997social}, every node holds an initial bias in opinion which could have  formed owing to 
the node's fundamental views, its experiences,
past information from news and other sources, opinion dynamics in the past, etc. 
%
This initial bias  plays an important role in determining a node's final opinion, and consequently the opinions of its neighbors and hence that of its neighbors' neighbors and so on. If nodes give significant weightage to their  biases, the camps would want to influence these biases. 
%
This could be implemented by campaigning in multiple phases or rounds,
 wherein the opinion at the conclusion of a phase would act as the initial biased opinion for the next phase. Such campaigning is often used during elections and  product marketing, so as to gradually drive the opinions of nodes.
With the possibility of campaigning in multiple phases, a camp could not only decide which nodes to invest on, but also how to split its available budget across different phases.


\vspace{-2mm}
\section{Related Work}
\label{sec:ODSNmultiphase_relevant}

The topic of opinion dynamics in social networks has received significant attention in the autonomous agents and multiagent systems community in recent times \cite{chatterjee2013predicting,crawford2013opposites,grandi2017strategic,grandi2015propositional,sina2015adapting,soriano2016simultaneous,tsang2014opinion,yadav2016using}.
Several models of opinion dynamics have been proposed and studied in the literature
\cite{acemoglu2011opinion,lorenz2007continuous}, 
%
%
%
%
some of the noteworthy models being
DeGroot \cite{degroot1974reaching}, Voter \cite{clifford1973model,holley1975ergodic}, DeGroot-Friedkin \cite{friedkin1990social,friedkin1997social}, bounded confidence \cite{krause2000discrete}, etc.
In DeGroot-Friedkin model, each node updates its opinion using a weighted convex combination of its initial opinion and neighbors' opinions. 
Our model can be viewed as a multiphase generalization of this model, while also accounting for 
the camps' investments.
%

%
Problems relating to maximizing opinion diffusion in social networks have been extensively studied in the literature \cite{guille2013information,easley2010networks,kempe2003maximizing}.
A primary task in such problems is to determine influential nodes,
which has been an important research area in the multiagent systems community 
 \cite{amoruso2017contrasting,cholvy2016influence,ghanem2012agents,li2017agent,maghami2012identifying,pasumarthi2015near}.
Gionis, Terzi, and Tsaparas \cite{gionis2013opinion} study the problem of identifying such nodes whose positive opinions 
would maximize the overall positive opinion 
in the  network.
%
%
Yildiz, Ozdaglar, and Acemoglu \cite{yildiz2013binary} study the problem of optimal placement of stubborn nodes (whose opinion values stay unchanged) in the discrete binary opinions setting, 
given the location of competing stubborn agents.
%
%
The competitive setting has resulted in several  game theoretic studies \cite{ghaderi2014opinion,anagnostopoulos2015competitive,bharathi2007competitive,goyal2014competitive}.
%
%
Specific to analytically tractable models such as DeGroot,
there have been studies in the competitive setting
 to identify influential nodes 
 and the amounts to be invested on them
\cite{dubey2006competing,bimpikis2016competitive,grabisch2017strategic}.
Our work extends these studies to multiple phases, that is, we aim to determine 
the influential nodes in different phases and how much they should be invested on in a given phase.

A few studies have incorporated the concept of multiple phases.
Singer \cite{singer2016influence} presents a survey of  adaptive seeding methods for influence maximization. 
%
%
Horel and Singer \cite{horel2015scalable} develop  such methods with provable guarantees for models such as Voter, where the influence of a set can be expressed as the sum of the influence of its constituent nodes. 
%
Dhamal, Prabuchandran, and Narahari \cite{dhamal2015multiphase,dhamal2016information} empirically study the problem of 
optimally splitting the  budget between two phases.
%
While the primary reasoning behind using multiple phases in these studies is adaptation of strategy based on observations, we aim to use multiple phases for influencing the initial biases of the nodes, which requires a very different treatment.





To the best of our knowledge, there has not been an analytical study on a rich model such as DeGroot-Friedkin, for opinion dynamics in multiple phases (not even for single camp).
Following are the specific contributions of this paper:



\begin{itemize} [leftmargin=5mm]

\item
We formulate the multi-phase objective function under DeGroot-Friedkin model, 
resulting in a zero-sum game.

\item
Focusing on two phases,
we derive Nash equilibrium investment strategies
and
the loss that a camp would incur if it played myopic strategy instead of farsighted.

\item
We go deeper into the analysis when a camp's influence on a node depends on its initial bias, and develop a polynomial time algorithm for determining an optimal way to split a camp's budget between the two phases.

\item 
With two competing camps, we show the existence of Nash equilibria under reasonable assumptions, and that they can be computed in polynomial time.


\end{itemize}

%
%
%
%
%

\vspace{-2mm}
\section{Model}
\label{sec:ODSNmultiphase_prob}
%
%
%

Given a social network,
let $N$ be the set of nodes and $E$ be the set of weighted  directed edges.
Let $n=|N|$, $m=|E|$.
The weights could hold any sign, since the influences could be positive or negative 
\cite{altafini2013consensus}.
The good camp (opinion value of $+1$) aims to maximize the sum (or average) of the opinion values of the nodes, while the bad camp (opinion value of $-1$) simultaneously aims to minimize it.

\vspace{-1mm}
\subsection{Parameters}
\label{sec:ODSNmultiphase_model}



\subsubsection{Initial Bias}

As per DeGroot-Friedkin model, every node $i$ holds an initial bias in opinion prior to the opinion dynamics process; we denote it by $v_i^0$. In the multiphase setting, this acts as the initial bias for the first phase.
We denote the opinion value of node $i$ at the conclusion of phase $p$ by $v_i^{(p)}$.
This acts as the initial opinion or bias for phase $p+1$.
Naturally, we can denote its initial bias for phase $p$ by $v_i^{(p-1)}$.
So by convention, we have $v_i^{(0)}=v_i^0$.
We denote the weightage that node $i$ attributes to its initial bias by $w_{ii}^0$.

\subsubsection{Network Effect}

Let $w_{ij}$ be the weightage attributed by node $i$ to the opinion of its connection $j$.
%
Consistent with DeGroot-Friedkin model, we assume the influence on node $i$ owing to node $j$ in phase $p$, to be $w_{ij}v_j^{(p)}$. So the net influence owing to all of its connections is $\sum_{j\neq i} w_{ij}v_j^{(p)}$. Let $w_{ii}$ denote the weightage attributed by node $i$ to its own current opinion. 
%
%


\subsubsection{Camp Investments}

The good and bad camps attempt to directly influence the nodes so that their opinions are pulled towards being positive and negative respectively. 
We denote the investments made by the good and bad camps on node $i$ by $x_i$ and $y_i$ respectively, and the weightage that node $i$ attributes to them by $w_{ig}$ and $w_{ib}$. Since the influence of good camp (opinion of $+1$) on node $i$ would be an increasing function of both $x_i$ and $w_{ig}$, we assume the influence to be $+w_{ig} x_i$ for simplicity. Similarly, $-w_{ib} y_i$ is the influence of bad camp (opinion of $-1$) on node $i$. 
%
Let $k_g$ and $k_b$ be the respective budgets of the good and bad camps, so the camps should invest such that $\sum_i x_i \leq k_g$ and $\sum_i y_i \leq k_b$.
%

\subsubsection{Matrix Forms}

Let $\mathbf{w}$ be the matrix consisting of the weights $w_{ij}$ for each pair of nodes $(i,j)$.
%
%
Let $\mathbf{v^0}$, $\mathbf{w^0}$, $\mathbf{w_g}$, $\mathbf{w_b}$, $\mathbf{x^{(p)}}$, $\mathbf{y^{(p)}}$, $\mathbf{v^{(p)}}$
be the vectors consisting of the elements
$v_i^0$, $w_{ii}^0$, $w_{ig}$, $w_{ib}$, $x_i^{(p)}$, $y_i^{(p)}$, $v_i^{(p)}$, respectively.
%
Let the operation $\circ$ denote Hadamard product 
of vectors, that is, $(\mathbf{a}\circ\mathbf{b})_{i} = {a}_i {b}_i$. 


\subsubsection{Weight Constraints}

Like standard opinion dynamics models, we have the following condition on influence weights:

\begin{small}
\begin{align*}
\forall i: \;
|w_{ii}^0|+|w_{ii}|+\sum_{j\neq i} |w_{ij}|+|w_{ig}|+|w_{ib}| \leq 1
\end{align*}
\end{small}

\noindent
A standard assumption for guaranteeing convergence is

\begin{small}
\begin{align*}
|w_{ii}|+\sum_{j\neq i} |w_{ij}| < 1
\end{align*}
\end{small}

\noindent
Since we would have weights attributed to influences external to the network (influences due to bias ($w_{ii}^0$) and campaigning ($w_{ig},w_{ib}$)),
this assumption is well suited for our model.


\subsection{Opinion Update Rule} 

The opinion update rule in DeGroot-Friedkin model is
$
\forall i:
v_i
\leftarrow w_{ii}^{0} v_i^0
+ w_{ii} v_i
+ \sum_{j \neq i} w_{ij}v_j
$,
or equivalently, 
$
\forall i:
v_i
\leftarrow w_{ii}^{0} v_i^0
+ \sum_{j} w_{ij}v_j
$.
Generalizing to multiphase setting, the update rule in the $p^{\text{th}}$ phase is
$
\forall i:
v_i^{(p)}
\leftarrow w_{ii}^{0} v_i^{(p-1)}
+ \sum_{j} w_{ij}v_j^{(p)}
$.
Taking the investments 
of the camps into account, it is
%
%

\begin{small}
\vspace{-2mm}
\begin{align}
\forall i:\;
v_i^{(p)}
&\leftarrow w_{ii}^{0} v_i^{(p-1)}
+ \sum_{j} w_{ij}v_j^{(p)}
+ w_{ig} x_i^{(p)}
- w_{ib} y_i^{(p)}
\nonumber
\end{align}
\vspace{-4mm}
\end{small}

%
\begin{small}
\vspace{-2mm}
\begin{align}
\Longleftrightarrow \;\;
\mathbf{v}^{\mathbf{(p)}} &\leftarrow  \mathbf{w^0} \! \circ \! \mathbf{v^{\mathbf{(p-1)}}} + \mathbf{w}\mathbf{v}^{\mathbf{(p)}}  + \mathbf{w_g} \! \circ \! \mathbf{x}^{\mathbf{(p)}} - \mathbf{w_b} \! \circ \! \mathbf{y}^{\mathbf{(p)}}
\label{eqn:update_rule_matrix}
\end{align}
\vspace{-3mm}
\end{small}

\noindent
The static vectors: $\mathbf{x}^{\mathbf{(p)}}, \mathbf{y}^{\mathbf{(p)}}, \mathbf{v^{\mathbf{(p-1)}}}$ (weighted by $\mathbf{w_g}, \mathbf{w_b}, \mathbf{w^0}$), are unchanged throughout a phase $p$, while
$\mathbf{v}^{\mathbf{(p)}}$ gets updated.






\noindent
Writing it as a recursion (with iterating variable $\tau$):
\begin{small}
\begin{align*}
\mathbf{v}_{\langle\tau\rangle}^{\mathbf{(p)}} &= \mathbf{w}\mathbf{v}_{\langle\tau-1\rangle}^{\mathbf{(p)}} + \mathbf{w^0}\circ\mathbf{v^{\mathbf{(p-1)}}} + \mathbf{w_g}\circ\mathbf{x}^{\mathbf{(p)}} - \mathbf{w_b}\circ\mathbf{y}^{\mathbf{(p)}}
\end{align*}
\vspace{-2mm}
\end{small}

\noindent
It can be simplified as
\begin{small}
\begin{align*}
\mathbf{v}_{\langle\tau\rangle}^{\mathbf{(p)}} \! &= \!
\mathbf{w}^\tau \mathbf{v}_{\langle 0 \rangle}^{\mathbf{(p)}}
 \! + \! \big( \sum_{\eta=0}^{\tau-1}  \mathbf{w}^\eta \big)(\mathbf{w^0} \! \circ \! \mathbf{v^{\mathbf{(p-1)}}} \! + \! \mathbf{w_g} \! \circ \! \mathbf{x}^{\mathbf{(p)}} \! - \! \mathbf{w_b} \! \circ \! \mathbf{y}^{\mathbf{(p)}}) 
\end{align*}
\end{small}

\noindent
Now $\mathbf{v}_{\langle 0 \rangle}^{\mathbf{(p)}}$ is the opinion at the beginning of phase $p$, which is $\mathbf{v}^{\mathbf{(p-1)}}$.
For determining if the process reaches convergence, we see the limiting value when $\tau \rightarrow \infty$. 


%

\noindent
Since $\mathbf{w}$ is a strictly substochastic matrix (with the sum of each of its rows strictly less than 1),
its spectral radius is less than 1.
So $\lim_{\tau \rightarrow \infty}\mathbf{w}^\tau = 0$, and
$\sum_{\eta=0}^\infty \mathbf{w}^\eta = (\mathbf{I}-\mathbf{w})^{-1}$
\cite{grabisch2017strategic}.

%
\noindent
Hence as $\tau \rightarrow \infty$, 
the above equation
can be written as

\begin{small}
\vspace{-3mm}
\begin{align*}
\lim_{\tau \rightarrow \infty}
\mathbf{v}_{\langle\tau\rangle}^{\mathbf{(p)}} = (\mathbf{I}-\mathbf{w})^{-1} (\mathbf{w^0}\circ\mathbf{v^{\mathbf{(p-1)}}} + \mathbf{w_g}\circ\mathbf{x}^{\mathbf{(p)}} - \mathbf{w_b}\circ\mathbf{y}^{\mathbf{(p)}})
\end{align*}
\vspace{-2mm}
\end{small}

\noindent
which is a constant vector.
%
%
%
So the  dynamics in phase $p$ converges to 
the steady state

\begin{small}
\vspace{-3mm}
\begin{align}
\hspace{-1.5mm}
\mathbf{v}^{\mathbf{(p)}} &= (\mathbf{I}-\mathbf{w})^{-1} (\mathbf{w^0}\circ\mathbf{v^{(p-1)}}  + \mathbf{w_g}\circ\mathbf{x^{(p)}} - \mathbf{w_b}\circ\mathbf{y^{(p)}})
\label{eqn:basic}
\end{align}
\vspace{-3mm}
\end{small}

\vspace{-5mm}
\section{Problem Formulation}
\label{sec:multiphase}


We now formulate the multi-phase objective function, and analyze the two-phase opinion dynamics in further detail.
We first derive an expression for  $\sum_i v_i^{(p)}$, the sum of the opinion values of the nodes in the network at the end of phase $p$.



\noindent
Premultiplying (\ref{eqn:basic}) by $\mathbf{1}^T$ gives

\begin{small}
\vspace{-3mm}
\begin{align*}
\mathbf{1}^T \mathbf{v}^{\mathbf{(p)}} &= \mathbf{1}^T (\mathbf{I}-\mathbf{w})^{-1} (\mathbf{w^0}\circ\mathbf{v^{(p-1)}} + \mathbf{w_g}\circ\mathbf{x^{(p)}} - \mathbf{w_b}\circ\mathbf{y^{(p)}})
\end{align*}
\vspace{-2mm}
\end{small}

\noindent
Let $\mathbf{r} = (\mathbf{I}-\mathbf{w}^T)^{-1} \mathbf{1} = ((\mathbf{I}-\mathbf{w})^{-1})^T \mathbf{1} = \left( \mathbf{1}^T (\mathbf{I}-\mathbf{w})^{-1} \right)^T$.
So the above is equivalent to 

\begin{small}
\vspace{-2mm}
\begin{align}
\sum_i v_i^{(p)} = \sum_i r_i w_{ii}^0 v_i^{(p-1)} + \sum_i r_i ( w_{ig}x_i^{(p)} - w_{ib}y_i^{(p)} ) 
\label{eqn:sumatP}
\end{align}
\vspace{-2mm}
\end{small}

\noindent
Premultiplying (\ref{eqn:basic}) by $(\mathbf{r}\circ\mathbf{w^0})^T$ gives

\begin{small}
\vspace{-2mm}
\begin{align*}
\hspace{-4mm}
(\mathbf{r} \! \circ \! \mathbf{w^0})^T \mathbf{v}^{\mathbf{(p)}} \! &= \! (\mathbf{r} \! \circ \! \mathbf{w^0})^T (\mathbf{I} \! - \! \mathbf{w})^{-1} (\mathbf{w^0} \! \circ \! \mathbf{v^{(p-1)}} \! + \! \mathbf{w_g} \! \circ \! \mathbf{x^{(p)}} \! - \! \mathbf{w_b} \! \circ \! \mathbf{y^{(p)}})
\end{align*}
\vspace{-2mm}
\end{small}

\noindent
Let the constant $(\mathbf{r}\circ\mathbf{w^0})^T \! (\mathbf{I}-\mathbf{w})^{-1} \! = \! \left( (\mathbf{I} \! - \! \mathbf{w}^T)^{-1} (\mathbf{r} \! \circ \! \mathbf{w^0}) \right)^T \! = \! \mathbf{s}^T$.
So the above is equivalent to 

\begin{small}
\vspace{-2mm}
\begin{align*}
\hspace{-2mm}
\sum_i r_i w_{ii}^0 v_i^{(p)} = \sum_i s_i w_{ii}^0 v_i^{(p-1)} + \sum_i s_i ( w_{ig}x_i^{(p)} - w_{ib}y_i^{(p)} ) 
\end{align*}
\end{small}
or equivalently,

\begin{small}
\vspace{-2mm}
\begin{align*}
\sum_i r_i w_{ii}^0 v_i^{(p-1)} &= \sum_i s_i w_{ii}^0 v_i^{(p-2)} + \sum_i s_i ( w_{ig}x_i^{(p-1)} - w_{ib}y_i^{(p-1)} ) 
\end{align*}
\end{small}
Substituting this in (\ref{eqn:sumatP}) gives
\begin{small}
\begin{align}
\nonumber
\sum_i v_i^{(p)} = \sum_i s_i w_{ii}^0 v_i^{(p-2)}  + \sum_i s_i ( w_{ig}x_i^{(p-1)} - w_{ib}y_i^{(p-1)} ) \\+ \sum_i r_i ( w_{ig}x_i^{(p)} - w_{ib}y_i^{(p)} ) 
\label{eqn:recurse}
\end{align}
\end{small}
Note that $\mathbf{s}$ can be viewed as a weighted extension of $\mathbf{r}$. If $(\mathbf{I}-\mathbf{w})^{-1} = \Delta$, then 
$r_i = \sum_j \Delta_{ji}$ and $s_i = \sum_j r_j w_{jj}^0 \Delta_{ji}$.



\textbf{Multiphase Katz centrality.}
%
$r_i = \left( (\mathbf{I}-\mathbf{w}^T)^{-1} \mathbf{1} \right)_i$ resembles Katz centrality of node $i$, capturing its influencing power on other nodes
if there are no subsequent phases.
However, while selecting optimal nodes for a phase if there are subsequent phases to follow, the effectiveness of  node $i$ depends on its influencing power over those nodes ($j$), which would give good weightage ($w_{jj}^0$) to their initial opinion in the next phase, as well as have a good influencing power over other nodes ($r_j$). This is captured by $s_i$, so it can be viewed as  two-phase  Katz centrality.
%
Let  vector $\mathbf{r}$ be denoted by $\mathbf{r}^{(1)}$ and $\mathbf{s}$ by $\mathbf{r}^{(2)}$. 
In general, for an integer $q>1$, let $r_i^{(q)} = \sum_j r_j^{(q-1)} w_{jj}^0 \Delta_{ji}$
or $\mathbf{r}^{(q)} = \Delta^T(\mathbf{r}^{(q-1)}\circ\mathbf{w^0}) $.
Intuitively, ${r}_i^{(q)}$ is the measure of influence of node $i$ looking $q$ phases ahead;
it can be viewed as  $q$-phase Katz centrality.
\qed

Now the recursion in Equation (\ref{eqn:recurse}) can be simplified as

\begin{small}
\vspace{-2mm}
\begin{align}
\nonumber
\hspace{-2mm}
\sum_i v_i^{(p)} = \sum_i {r}_i^{(2)} w_{ii}^0 v_i^{(p-2)}  + \sum_i {r}_i^{(2)} ( w_{ig}x_i^{(p-1)} - w_{ib}y_i^{(p-1)} )  \\ + \sum_i {r}_i^{(1)} ( w_{ig}x_i^{(p)} - w_{ib}y_i^{(p)} ) 
\nonumber
\end{align}
\vspace{-2mm}
\end{small}

\noindent
Solving the above recursion to eliminate $v_i^{(p-2)}$ gives

\begin{small}
\begin{align}
\hspace{-3mm}
\sum_i v_i^{(p)} \! &= \! \sum_i r_i^{(p)} w_{ii}^0 v_i^0  + \sum_{q=1}^{p} \sum_i r_i^{(p-q+1)} ( w_{ig}x_i^{(q)} - w_{ib}y_i^{(q)} )
\label{eqn:multiphase}
\end{align}
\end{small}

\noindent
which is now a function of  intrinsic constants and strategies.
%
\\
If there is only a single phase, we have

\begin{small}
\begin{align}
\sum_i v_i^{(1)} &= \sum_i r_i w_{ii}^0 v_i^0  + \sum_i r_i w_{ig}x_i^{(1)} - \sum_i r_i w_{ib}y_i^{(1)} 
\label{eqn:singlephase}
\end{align}
\end{small}
We focus on the two-phase campaign in our study.
\begin{small}
\begin{align}
\nonumber
\sum_i v_i^{(2)} = \sum_i s_i w_{ii}^0 v_i^{0} + \sum_i s_i w_{ig}x_i^{(1)} - \sum_i s_i  w_{ib}y_i^{(1)}  \\  + \sum_i r_i  w_{ig}x_i^{(2)} - \sum_i r_i  w_{ib}y_i^{(2)} 
\label{eqn:2phase}
\end{align}
\end{small}

\noindent
Here $(\mathbf{x^{(1)}},\mathbf{x^{(2)}})$ is the strategy of  good camp, and $(\mathbf{y^{(1)}},\mathbf{y^{(2)}})$ is the strategy of  bad camp.
Given an investment strategy profile $\big((\mathbf{x^{(1)}},\mathbf{x^{(2)}}),(\mathbf{y^{(1)}},\mathbf{y^{(2)}})\big)$,
let $u_g\big((\mathbf{x^{(1)}},\mathbf{x^{(2)}}),(\mathbf{y^{(1)}},\mathbf{y^{(2)}})\big)$ 
be  the good camp's utility
and $u_b\big((\mathbf{x^{(1)}},\mathbf{x^{(2)}}),(\mathbf{y^{(1)}},\mathbf{y^{(2)}})\big)$ be the bad camp's utility.
The good camp aims to maximize the value of  (\ref{eqn:2phase}), while the bad camp aims to minimize it. So,

\begin{small}
\vspace{-3mm}
\begin{gather}
\nonumber
u_g\big((\mathbf{x^{(1)}},\mathbf{x^{(2)}}),(\mathbf{y^{(1)}},\mathbf{y^{(2)}})\big)= \sum_i v_i^{(2)}
\\
\;\;\; \text{and} \;\;\;
u_b\big((\mathbf{x^{(1)}},\mathbf{x^{(2)}}),(\mathbf{y^{(1)}},\mathbf{y^{(2)}})\big)= -\sum_i v_i^{(2)}
\label{eqn:2phaseutilities}
\end{gather}
\vspace{-3mm}
\end{small}

\noindent
with the following constraints on the investment strategies:

\begin{small}
\vspace{-2mm}
\begin{gather}
\nonumber
\sum_i \big( x_i^{(1)} + x_i^{(2)} \big) \leq k_g \;\;,\;\; \sum_i \big( y_i^{(1)} + y_i^{(2)} \big) \leq k_b
\\
\nonumber
\forall i: x_i^{(1)} , x_i^{(2)} , y_i^{(1)} , y_i^{(2)} \geq 0
\end{gather}
\vspace{-3mm}
\end{small}

\noindent
The game can thus be viewed as a two-player zero-sum game,
where the players determine their investment strategies $(\mathbf{x^{(1)}},\mathbf{x^{(2)}})$ and $(\mathbf{y^{(1)}},\mathbf{y^{(2)}})$;
the good camp invests as per $\mathbf{x^{(1)}}$ in the first phase and as per $\mathbf{x^{(2)}}$ in the second phase, and 
the bad camp invests as per $\mathbf{y^{(1)}}$ in the first phase and as per $\mathbf{y^{(2)}}$ in the second phase.
Our objective essentially is to find the Nash equilibrium strategies of the two camps.

%

\vspace{-2mm}
\section{
Multiphase 
Strategies
}

We call a camp farsighted if it computes its strategy considering that there would be a second phase, that is, it considers the objective function $\sum_i v_i^{(2)}$ in Equation (\ref{eqn:2phase}). 
A camp can be called myopic if it computes its strategy greedily by considering the near-sighted objective function $\sum_i v_i^{(1)}$.
We now consider cases where both camps are farsighted or myopic, and where one camp is farsighted while the other is myopic. 

%



When both camps are farsighted, the game is as per (\ref{eqn:2phaseutilities}).
%
%
Here, the optimal strategy of the good camp is to invest the entire budget on node $i$ with the maximum value of $\max \left\{ s_i w_{ig} , r_i w_{ig} \right\}$ if this value is positive (no investment is made if this value is non-positive). If $\max_i s_i w_{ig} > \max_j r_j w_{jg}$, the entire budget is invested in the first phase, while if $\max_i s_i w_{ig} < \max_j r_j w_{jg}$, the entire budget is invested in the second phase (indifferent between the phases if $\max_i s_i w_{ig} = \max_j r_j w_{jg}$).
The optimal strategy of the bad camp is analogous.
%
%
Furthermore, the optimal strategies of the camps are mutually independent, since the optimization terms with respect to different variables are decoupled.



When one camp (say the bad camp) is myopic, while the other camp (say the good camp) is farsighted,
the optimal strategy of the good camp is the same as above since it is independent of the strategy of the bad camp. 
Though the actual utility of the bad camp is $-\sum_i v_i^{(2)}$,
it perceives its utility as
$-\sum_i v_i^{(1)}$ (see Equation (\ref{eqn:singlephase})), owing to its myopic nature.
The optimal strategy of the bad camp is hence to invest the entire budget on node $i$ with the maximum value of $r_i w_{ib}$ in the first phase itself (if positive).
The loss to the bad camp owing to its myopic play is thus,

\begin{small}
\vspace{-3mm}
\begin{align*}
k_b \Big( 
\max_i \left( \max \left\{ s_i w_{ib} , r_i w_{ib} , 0 \right\} \right) - \max \left\{ s_{\hat{i}} w_{\hat{i}b}  , 0 \right\}
\Big) 
\end{align*}
\vspace{-2mm}
\end{small}

\noindent
\text{where} $\hat{i} = \argmax_i r_i w_{ib}$.

Note that when both camps are myopic, their strategies are as per the  single phase objective function (\ref{eqn:singlephase}), that is, the good camp invests entirely on a node $i$ with the highest value of $r_i w_{ig}$ (if positive) and the bad camp invests entirely on a node $i$ with the highest value of $r_i w_{ib}$ (if positive) in the first phase.
However, the actual utility would be as per the two phase objective function (\ref{eqn:2phase}) with $x_i^{(2)}=y_i^{(2)}=0,\forall i$.

%
%
%
%
%
%
%
%
%
%
%
%
%



\textbf{Case of Bounded Investment per Node.}
Here we consider the case when there is a bound on investment on every node in either phase, for instance, $0 \leq x_i^{(1)},x_i^{(2)},y_i^{(1)},y_i^{(2)} \leq 1, \forall i$.
%
%
Such bounded investments could be owing to 
discounts which cannot exceed 100\%, 
attention capacity or time constraint of a voter to receive convincing arguments, 
company policy to limit the number of free samples that can be given to a customer, 
government policy of limiting the monetary investment by a camp on a voter, 
etc.
Bounded investments would ensure that the opinion values are bounded (from Equation (\ref{eqn:basic})).
Bounded opinion values can be associated with voting in elections and product adoption, where a bounded value could be transformed into a probability of voting for a party or adopting a product. For instance, an opinion value of $v$ could imply that the probability of voting for the good camp (or adopting the good camp's product) is $(1+v)/2$ and that of voting for the bad camp is $(1-v)/2$.
The sum of opinion values would thus imply the expected number of votes in the favor of the good camp.


%
%
When both camps are farsighted,
for computing optimal strategy of the good camp, we consider the set of elements $\{s_i w_{ig},r_i w_{ig}\}_{i\in N}$. Until the budget $k_g$ is exhausted, we choose an element with the maximum value one at a time, which is not already chosen. If the chosen element is $s_j w_{jg}$, we invest on node $j$ in the first phase, while if the chosen element is $r_j w_{jg}$, we invest on node $j$ in the second phase, subject to a maximum investment of 1 unit.
The strategy of the bad camp is analogous.
Note again that the optimal strategies of the camps are mutually independent.

The scenario of one camp (say the bad camp) being myopic can be analyzed on similar lines as in the case of no bounds on investment per node, with the exception that instead of the bad camp investing entirely in one node with the highest value of $r_i w_{ib}$ in the first phase, it would invest on nodes according to decreasing values of $r_i w_{ib}$, subject to a maximum investment of 1 unit per node.
Hence the loss to the bad camp can be determined.
The scenario of both camps being myopic can also be analyzed analogously.


\begin{remark}
When there were no bounds on investment per node, it was optimal for a farsighted camp to invest its entire budget on one node in one phase. The case with the bounds on investment per node requires a farsighted camp to invest on different nodes, possibly in different phases. Such a strategy would result in an outcome in which the camps could split their budget across the two phases. 
Specifically, the budget allocated by the good camp for the first and second phases would be $\sum_i x_i^{(1)}$ and $\sum_i x_i^{(2)}$ respectively, while that allocated by the bad camp would be $\sum_i y_i^{(1)}$ and $\sum_i y_i^{(2)}$.
\end{remark}
%
\textbf{Extension to Multiple Elections Setting.}
The multi-phase study can be directly extended to the multiple elections setting, where the outcome of the first election is in itself valuable. If the immediate outcome holds a weightage of $\delta_1$ and the future outcome holds a weightage of $\delta_2$, the objective function now transforms into $ \sum_i (\delta_1 v^{(1)} + \delta_2 v^{(2)} )$. The decision parameter $s_i$ would then get replaced by $\delta_1 r_i + \delta_2 s_i$.

\vspace{-2mm}
\section{
Dependency of $\mathbf{w_g}$ and $\mathbf{w_b}$ on Initial Bias}
\label{sec:dependency_basic}

So far, we considered a setting where $w_{ig}$ and $w_{ib}$ for a node $i$, could be any fixed values. However, in several scenarios, 
more often that not, the initial opinion of a node would impact the effectiveness of the investment made by the two camps.
For example, in an election setting, if the initial opinion of a node is highly positive, it is likely to be the case that
 the investment made by the good camp would be more effective than that made by the bad camp. 
 In other words, such a node is likely to be more easily influenced by (or is more sensitive towards)  the investment made by the good camp.
 That is, $v_i^{(p-1)}>0$ would likely mean $w_{ig}^{(p)} > w_{ib}^{(p)}$.
%
The fundamental reasoning is similar to that of the bounded confidence model \cite{krause2000discrete}, where a node pays more attention to opinions that do not differ too much from its own opinion.

Note also that $w_{ii}^0$ would also play a role, since it says how much weightage is given by node $i$ to its initial opinion. 
So a theoretical model can be proposed on this line, wherein $w_{ig}^{(p)}$ is a monotone non-decreasing function of $w_{ii}^0 v_i^{(p-1)}$, and $w_{ib}^{(p)}$ is a monotone non-increasing function of $w_{ii}^0 v_i^{(p-1)}$.
Let node $i$ attribute a total of $\theta_i$ to the influence weights of the camps
(that is, $w_{ig}^{(p)}+w_{ib}^{(p)}=\theta_i$),
and ${\Theta}$
be the vector consisting of {elements
$\theta_i$.}
So one such simple linear model could be

\begin{small}
\vspace{-4mm}
\begin{align}
w_{ig}^{(p)} = \theta_i \left( \frac{1+w_{ii}^0 v_i^{(p-1)}}{2} \right)
\text{ , }
w_{ib}^{(p)} = \theta_i \left( \frac{1-w_{ii}^0 v_i^{(p-1)}}{2} \right)
\label{eqn:WonV}
\end{align}
\end{small}
\vspace{-3mm}


\label{sec:dependency_2phase}

The dependency of $w_{ig}$ and $w_{ib}$ on the initial opinion
is particularly significant in multi-phase campaigning, since the initial opinion changes across phases, and hence the assumption of $w_{ig}$ and $w_{ib}$ being the same across phases, may not be valid (since these values now depend on $v_i^{(p-1)}$ which would in general be different for different phases).
%
Considering that $w_{ig}$ and $w_{ib}$ can be different across phases, Equation (\ref{eqn:multiphase}) can be generalized to

\begin{small}
\vspace{-2mm}
\begin{align}
\nonumber
\sum_i v_i^{(p)} &= \sum_i r_i^{(p)} w_{ii}^0 v_i^0  + \sum_{q=1}^{p} \sum_i r_i^{(p-q+1)} ( w_{ig}^{(q)} x_i^{(q)} - w_{ib}^{(q)} y_i^{(q)} )
\end{align}
\end{small}
\vspace{-1mm}


\noindent
We now focus on two-phase campaign ($p=2$).
Recall that $\Delta = (\mathbf{I}-\mathbf{w})^{-1}$.
We now derive the objective function $\sum_i v_i^{(2)}$.

\begin{small}
\begin{align}
&\sum_j v_j^{(2)} 
= \sum_j r_j ( w_{jj}^0 v_j^{(1)} + w_{jg}^{(2)} x_j^{(2)} - w_{jb}^{(2)} y_j^{(2)} )
\nonumber
\\
&= \sum_j r_j \left( w_{jj}^0 v_j^{(1)} + \frac{\theta_j}{2} ( {1+w_{jj}^0 v_j^{(1)}} ) x_j^{(2)} -  \frac{\theta_j}{2} ( {1-w_{jj}^0 v_j^{(1)}} ) y_j^{(2)} \right)
\nonumber
\\
&= \sum_j r_j w_{jj}^0 v_j^{(1)} \left( 1 + \frac{\theta_j}{2} x_j^{(2)} + \frac{\theta_j}{2} y_j^{(2)} \right) + \sum_j r_j \frac{\theta_j}{2} (x_j^{(2)} - y_j^{(2)})
\nonumber
\end{align}
\end{small}

\noindent
The first term $\sum_j r_j w_{jj}^0 v_j^{(1)} \left( 1 + \frac{\theta_j}{2} x_j^{(2)} + \frac{\theta_j}{2} y_j^{(2)} \right)$
 can be obtained (without $v_j^{(1)}$)
by 
premultiplying (\ref{eqn:basic}) by $(\mathbf{r} \circ \mathbf{w^0} \circ (\mathbf{1} + \frac{\Theta}{2}\circ \mathbf{x^{(2)}} + \frac{\Theta}{2}\circ \mathbf{y^{(2)}}))^T$ and using (\ref{eqn:WonV}).
Hence we get

\begin{small}
\begin{align}
\hspace{-2mm}
\nonumber
\sum_i v_i^{(2)} &=
\sum_i \sum_j  \left( w_{ii}^0 v_i^{0}
\Big(1+\frac{\theta_i}{2} x_i^{(1)} + \frac{\theta_i}{2} y_i^{(1)} \Big) + \frac{\theta_i}{2}(x_i^{(1)}-y_i^{(1)}) \right)
\\\hspace{-2mm} &
\left( r_j w_{jj}^0 \Delta_{ji}  \Big(1+\frac{\theta_j}{2} x_j^{(2)} + \frac{\theta_j}{2} y_j^{(2)}\Big) \right)
+ \sum_j r_j \frac{\theta_j}{2} (x_j^{(2)} - y_j^{(2)})
\label{eqn:orig_dep_2camps}
\end{align}
\vspace{-4mm}
\end{small}

Note that unlike the 
previous setting
(Equation (\ref{eqn:2phase})), the optimal strategies of the two camps in this setting would be mutually dependent in general,
since the optimization terms with respect to different variables are coupled.

\vspace{-2mm}
\subsection{The Case of a Single Camp}
\label{sec:dep_onecamp}


First we consider a simplified yet interesting version by considering only one camp (the good camp).
%
Here we would have $y_i^{(1)}=y_i^{(2)}=0, \forall i$ in
Expression (\ref{eqn:orig_dep_2camps}).
%
%
For simplicity of notation, denote $r_j w_{jj}^0 \Delta_{ji}= b_{ji}$ and $w_{ii}^0 v_i^{0} = c_i$. 
So we have

\begin{small}
\vspace{-1mm}
\begin{align}
\hspace{-5mm}
\nonumber
\sum_i v_i^{(2)}  
 &= 
\sum_i \sum_j \! \left( c_i
\! + \! \frac{\theta_i}{2}x_i^{(1)} (c_i \! + \! 1) \right) \!
\left( b_{ji}  \Big(1 \! + \! \frac{\theta_j}{2} x_j^{(2)} \Big) \right)
\\
&\;\;\;\;\;\;\;\;\;\;\;\;\;\;\;\;\;\;\;\;\;\;\;\;\;\;\;\;\;\;\;\;\;\;\;\;\;\;\;\;\;\;\;\;\;\;\;\;\;\;\;\;\;\;\;
\! + \! \sum_j r_j \frac{\theta_j}{2} x_j^{(2)} 
\label{eqn:raw2phase}
\\
\nonumber
&= 
\sum_i x_i^{(1)} \bigg( \frac{\theta_i}{2} (c_i \! + \! 1) \sum_j 
 b_{ji}  \Big(1 \! + \! \frac{\theta_j}{2} x_j^{(2)} \Big) \bigg)
\\&\;\;\;\;\;\;\;\;\;
+
\sum_i \sum_j c_i
 b_{ji}  \bigg(1+\frac{\theta_j}{2} x_j^{(2)} \bigg)  + \sum_j r_j \frac{\theta_j}{2} x_j^{(2)} 
 \label{eqn:raw2phasefix2}
 \\
 \nonumber
 &= 
 \sum_j x_j^{(2)}
  \left( \frac{\theta_j}{2} \sum_i b_{ji} \Big( c_i + \frac{\theta_i}{2}x_i^{(1)} (c_i+1) \Big) + \frac{\theta_j}{2}  r_j  \right)
  \\&\;\;\;\;\;\;\;\;\;\;\;\;\;\;\;\;\;\;\;\;\;\;\;\;
  +
  \sum_i \sum_j b_{ji} \bigg( c_i
    + \frac{\theta_i}{2}x_i^{(1)} (c_i+1) \bigg)
     \label{eqn:raw2phasefix1}
\end{align}
\vspace{-1mm}
\end{small}

\noindent
Deducing from Equations (\ref{eqn:raw2phasefix2}) and (\ref{eqn:raw2phasefix1}) that the above objective function is a bilinear function in $\mathbf{x^{(1)}}$ and $\mathbf{x^{(2)}}$, we prove the following claim.


\begin{lemma}
It is optimal to exhaust the entire budget $(k_g^{(1)}+k_g^{(2)}=k_g)$, and if not, it is optimal to not invest at all $(k_g^{(1)}=k_g^{(2)}=0)$.
Furthermore, it is an optimal strategy to invest on at most one node in a given phase.
\label{lem:allornothing}
\end{lemma}
\begin{proof}
Given any $\mathbf{x^{(2)}}$, Expression (\ref{eqn:raw2phasefix2}) can be maximized with respect to $\mathbf{x^{(1)}}$ by allocating $k_g - \sum_j x_j^{(2)}$ to a single node $i$ that maximizes
%
$
\frac{\theta_i}{2} (c_i+1) \sum_j 
 b_{ji}  \left(1+\frac{\theta_j}{2} x_j^{(2)} \right) 
 $,
  if this value is positive.
  In case of multiple such nodes, one node can be chosen at random.
  If this value is non-positive for all the nodes, it is optimal to have $\mathbf{x^{(1)}} = \mathbf{0}$.
  When $\mathbf{x^{(1)}} = \mathbf{0}$, 
  Expression (\ref{eqn:raw2phasefix1}) now implies that it is optimal to allocate the entire budget $k_g$ in the second phase to a single node $j$ that maximizes
$
    \frac{\theta_j}{2} \left( \sum_i b_{ji} c_i + r_j \right)
$,
      if this value is positive. 
      If this value is non-positive for all the nodes, it is optimal to have $\mathbf{x^{(2)}} = \mathbf{0}$.
      This is the case where starting with an $\mathbf{x^{(2)}}$, we arrive at the conclusion that it is optimal to either invest $k_g - \sum_j x_j^{(2)}$ on a single node in the first phase, or invest the entire budget $k_g$ on a single node in the second phase, or invest in neither of the phases.
      
      Similarly using (\ref{eqn:raw2phasefix1}),
starting with a given $\mathbf{x^{(1)}}$, we can conclude that it is optimal to either invest $k_g - \sum_j x_j^{(1)}$ on a single node in the second phase, or invest the entire $k_g$ on a single node in the first phase, or invest in neither phase.

So starting from any $\mathbf{x^{(1)}}$ or $\mathbf{x^{(2)}}$, we can iteratively improve (need not be strictly) on the value of (\ref{eqn:raw2phase}) by investing on at most one node in a given phase.
Furthermore, it is suboptimal to have $k_g^{(1)}+k_g^{(2)}<k_g$ unless $k_g^{(1)}=k_g^{(2)}=0$. 
\end{proof}


%
%
So there exist optimal vectors $\mathbf{x^{(1)}}$ and $\mathbf{x^{(2)}}$ that maximize (\ref{eqn:raw2phase}), such that $x_{\alpha}^{(1)} = k_g^{(1)}, x_{\beta}^{(2)} = k_g^{(2)} , x_{i \neq \alpha}^{(1)} = x_{j \neq \beta}^{(2)} = 0$.
Now the next step is to find the nodes $\alpha$ and $\beta$ that maximize (\ref{eqn:raw2phase}).
By incorporating $\alpha$ and $\beta$, the expression for $\sum_i v_i^{(2)}$ can be rewritten as


\begin{small}
\vspace{-1mm}
\begin{align}
\hspace{-5mm}
\nonumber
&
\sum_{i\neq \alpha} \sum_{j\neq \beta} c_i b_{ji} \! + \!
 \left( c_{\alpha}
\! + \! \frac{\theta_{\alpha}}{2}k_g^{(1)} (c_{\alpha} \! + \! 1) \right)
\left( b_{{\beta}{\alpha}}  \Big(1 \! + \! \frac{\theta_{\beta}}{2} k_g^{(2)} \Big) \right)
\! + \! r_{\beta} \frac{\theta_{\beta}}{2} k_g^{(2)} 
\\
\hspace{-5mm}
\nonumber
&\;\;\;\;\;\;
+ \sum_{j\neq \beta} b_{{j}{\alpha}} \left( c_{\alpha}
+ \frac{\theta_{\alpha}}{2}k_g^{(1)} (c_{\alpha}+1) \right)
+ \sum_{i\neq \alpha}  c_i
\left( b_{{\beta}i}  \Big(1+\frac{\theta_{\beta}}{2} k_g^{(2)} \Big) \right)
\\
\hspace{-5mm}
\nonumber
&=
\sum_i \sum_j c_i b_{ji} + \frac{\theta_{\alpha}}{2} k_g^{(1)}  \sum_j b_{j\alpha}(c_{\alpha}+1) 
+ \frac{\theta_{\beta}}{2} k_g^{(2)}  \Big(\sum_i b_{\beta i}c_i + r_{\beta}\Big)
\\
\hspace{-5mm} &\;\;\;\;\;\;\;\;\;\;\;\;\;\;\;\;\;\;\;\;\;\;\;\;\;\;\;\;\;\;\;\;\;\;\;\;\;\;\;\;\;\;\;\;
+\frac{\theta_{\alpha}\theta_{\beta}}{4} k_g^{(1)}  k_g^{(2)} b_{\beta \alpha} (c_{\alpha}+1)
\label{eqn:twophase_onecamp_1}
\end{align}
\end{small}
\vspace{-1mm}
%

Now, for a given pair $(\alpha,\beta)$, we will find the optimal values of $k_g^{(1)}$ and $k_g^{(2)}$ from (\ref{eqn:twophase_onecamp_1}). By definition, we have $k_g^{(2)} = k_g - k_g^{(1)}$. So the expression to be maximized is

\begin{small}
\begin{align*}
\hspace{-5mm}
\sum_i \sum_j c_i b_{ji} +
 k_g^{(1)} \frac{\theta_{\alpha}}{2} \sum_j b_{j\alpha}(c_{\alpha} \! + \! 1) 
+ (k_g \! - \! k_g^{(1)}) \frac{\theta_{\beta}}{2} \Big(\sum_i b_{\beta i}c_i \! + \! r_{\beta}\Big)
\\ +k_g^{(1)}  (k_g-k_g^{(1)}) \frac{\theta_{\alpha}\theta_{\beta}}{4} b_{\beta \alpha} (c_{\alpha}+1)
\end{align*}
\end{small}

\noindent
Equating the first derivative of the above expression with respect to $k_g^{(1)}$ to zero, we get

\begin{small}
\begin{align*}
k_g^{(1)} = \frac{k_g}{2} + \frac{\sum_j b_{j\alpha}}{\theta_{\beta}b_{\beta \alpha}} - \frac{\sum_i b_{\beta i}c_i + r_{\beta}}{\theta_{\alpha}b_{\beta \alpha}(c_{\alpha}+1)}
\end{align*}
\end{small}

\noindent
A valid value of $k_g^{(1)}$ can be obtained only if 
the denominators in the above expression are non-zero. 
However, note that a zero denominator would mean that Expression (\ref{eqn:twophase_onecamp_1}) is linear, resulting in only two possibilities of $k_g^{(1)}$, namely, $0$ or $k_g$.
%
%
Also, if the second derivative with respect to $k_g^{(1)}$ is positive, that is, $- \theta_{\alpha}\theta_{\beta} r_\beta w_{\beta\beta}^0 \Delta_{\beta\alpha} (w_{\alpha\alpha}^0 v_\alpha^{0}+1) > 0$, the optimal $k_g^{(1)}$ could be either $0$ or $k_g$.

If the second derivative with respect to $k_g^{(1)}$ is negative, that is, $- \theta_{\alpha}\theta_{\beta} b_{\beta \alpha} (c_{\alpha}+1) = - \theta_{\alpha}\theta_{\beta} r_\beta w_{\beta\beta}^0 \Delta_{\beta\alpha} (w_{\alpha\alpha}^0 v_\alpha^{0}+1) < 0$,
and since $k_g^{(1)}$ is bounded between $0$ and $k_g$, the optimal $k_g^{(1)}$ for a given pair $(\alpha,\beta)$ is
(since $b_{ji}=r_j w_{jj}^0 \Delta_{ji}$, $c_i=w_{ii}^0 v_i^{0}$),


\begin{small}
\begin{align}
\min \! \left\{ \! \max \! \bigg\{  \frac{k_g}{2} \! + \! \frac{s_{\alpha}}{\theta_{\beta} r_{\beta} w_{\beta \beta}^0 \Delta_{\beta \alpha}} \! - \! \frac{1+w_{\beta \beta}^0 \sum_i  \Delta_{\beta i} w_{ii}^0 v_i^{0}}{\theta_{\alpha}  w_{\beta \beta}^0 \Delta_{\beta \alpha} (1 \! + \! w_{\alpha \alpha}^0 v_{\alpha}^{0})} , 0  \bigg\}  ,   k_g \! \right\}
\label{eqn:optk1_new}
\end{align}
\end{small}

\noindent
and the corresponding optimal value of $k_g^{(2)}$ is

\begin{small}
\begin{align}
\min \! \left\{ \! \max \! \bigg\{ \frac{k_g}{2} \! - \! \frac{s_{\alpha}}{\theta_{\beta} r_{\beta} w_{\beta \beta}^0 \Delta_{\beta \alpha}} \! + \! \frac{1+w_{\beta \beta}^0 \sum_i  \Delta_{\beta i} w_{ii}^0 v_i^{0}}{\theta_{\alpha}  w_{\beta \beta}^0 \Delta_{\beta \alpha} (1 \! + \! w_{\alpha \alpha}^0 v_{\alpha}^{0})} , 0  \bigg\} , k_g \! \right\}
\label{eqn:optk2_new}
\end{align}
\end{small}
%
%
%

When we assumed $k_g^{(1)}$ and $k_g^{(2)}$ to be fixed, we had to iterate through all $(\alpha,\beta)$ pairs to determine the one that gives the optimal value of Expression (\ref{eqn:twophase_onecamp_1}). Now, whenever we look at an $(\alpha,\beta)$ pair, we can determine the corresponding optimal values of $k_g^{(1)}$ and $k_g^{(2)}$ using (\ref{eqn:optk1_new}) and (\ref{eqn:optk2_new}), and hence determine the value of Expression (\ref{eqn:twophase_onecamp_1}) by plugging in the optimal $k_g^{(1)}$ and $k_g^{(2)}$ and that $(\alpha,\beta)$ pair. 
 %
The optimal pair $(\alpha,\beta)$ can thus be obtained as the pair that maximizes (\ref{eqn:twophase_onecamp_1}).

The above analysis holds when $k_g^{(1)}+k_g^{(2)}=k_g$.
From Lemma~\ref{lem:allornothing}, we need to consider one more possibility that 
$k_g^{(1)}=k_g^{(2)}=0$, which gives a constant value $\sum_i \sum_j c_i b_{ji}$ for Expression (\ref{eqn:twophase_onecamp_1}).
Let $(0,0)$ correspond to this additional possibility.
%
%
It is hence optimal to invest $k_g^{(1)}$ (obtained using (\ref{eqn:optk1_new})) on node $\alpha$ in the first phase and $k_g^{(2)}$ (obtained using (\ref{eqn:optk2_new})) on node $\beta$ in the second phase, subject to it giving a  value greater than $\sum_i \sum_j c_i b_{ji}$ to Expression (\ref{eqn:twophase_onecamp_1}).
%


Since we iterate through $(n^2+1)$ possibilities (namely, $(\alpha,\beta) \in N \times N \cup \{(0,0)\}$), the above procedure gives a polynomial time algorithm for determining the optimal budget split and the optimal investments on nodes in two phases.

\begin{remark}
Expression (\ref{eqn:optk1_new}) indicates that for a given $(\alpha,\beta)$ pair, the good camp would want to invest more in the first phase for a higher $s_\alpha$. This is intuitive from our understanding of $s_\alpha$ being viewed as the Katz centrality of node $\alpha$ looking two phases ahead. Similarly, Expression (\ref{eqn:optk2_new}) indicates that it would want to invest more in the second phase for a higher $r_\beta$, since $r_\beta$ can be viewed as the Katz centrality of node $\beta$ looking just one phase ahead. 
Also as per (\ref{eqn:WonV}), $w_{ig}$ is an increasing function of $\theta_i$.
Now taking cue from Expression (\ref{eqn:2phase}), a node with a higher $w_{ig}$ is likely to be invested on more by the good camp. A similar insight is given by Expressions (\ref{eqn:optk1_new}) and (\ref{eqn:optk2_new}), where a higher $\theta_\alpha$ drives the camp to invest in the first phase and a higher $\theta_\beta$ drives it to invest in the second phase.
We further attempt to illustrate the role of $w_{ii}^0$ using simulations in Section \ref{sec:ODSNmultiphase_sim}.
\end{remark}

\vspace{-2mm}
\subsection{The Case of Competing Camps}
We now analyze the dependency setting with two camps.
%

Similar to the proof of Lemma \ref{lem:allornothing},
we can show that $\sum_i v_i^{(2)}$ is a multilinear function, since it can be written as a linear function in $\mathbf{x^{(1)}},\mathbf{y^{(1)}},\mathbf{x^{(2)}},\mathbf{y^{(2)}}$ separately.
The following lemma can be hence proved on similar lines.

\begin{lemma}
For either camp, it is optimal to exhaust the entire budget, and if not, it is optimal to not invest at all.
Furthermore, it is an optimal strategy to invest on at most one node in a given phase.
\label{lem:allornothing2camps}
\end{lemma}

Denoting $r_j w_{jj}^0 \Delta_{ji}= b_{ji}$ and $w_{ii}^0 v_i^{0} = c_i$ as before,
from Expression (\ref{eqn:orig_dep_2camps}),
$\sum_i v_i^{(2)}$ can be expanded as

\begin{small}
\begin{align}
\hspace{-5mm}
\nonumber 
& \sum_i \sum_j c_i b_{ji} 
\! + \! \sum_j x_j^{(2)} \frac{\theta_j}{2} \Big( \sum_i  c_i b_{j i} \! + \! r_j \Big)
\! + \! \sum_j y_j^{(2)} \frac{\theta_j}{2} \Big( \sum_i c_i b_{j i} \! - \! r_j \Big)
\\ \hspace{-5mm}&\;\;\;\;
\nonumber
+ \sum_i x_i^{(1)}  \frac{\theta_i}{2} (1 \! + \! c_i) \Big( s_{i} 
\! + \! \sum_j x_j^{(2)} \frac{\theta_j}{2} b_{j i}
\! + \! \sum_j y_j^{(2)} \frac{\theta_j}{2} b_{j i} \Big)
\\ \hspace{-5mm}&\;\;\;\;
- \sum_i y_i^{(1)} \frac{\theta_i}{2} (1 \! - \! c_i) \Big( s_i 
\! + \! \sum_j x_j^{(2)} \frac{\theta_j}{2} b_{j i}
\! + \! \sum_j y_j^{(2)} \frac{\theta_j}{2} b_{j i} \Big)
\label{eqn:raw2phase2camp}
\end{align}
\end{small}
From Lemma \ref{lem:allornothing2camps}, there exist optimal vectors $\mathbf{x^{(1)}},\mathbf{x^{(2)}}$ for  good camp and optimal vectors $\mathbf{y^{(1)}},\mathbf{y^{(2)}}$ for  bad camp, such that $x_{\alpha}^{(1)} = k_g^{(1)}, x_{\beta}^{(2)} = k_g^{(2)} , y_{\gamma}^{(1)} = k_b^{(1)} , y_{\delta}^{(2)} = k_b^{(2)}$, and $x_{i \neq \alpha}^{(1)} = x_{j \neq \beta}^{(2)} = y_{i \neq \gamma}^{(1)} = y_{j \neq \delta}^{(2)} = 0$.
Assuming such profile of nodes $((\alpha,\beta),(\gamma,\delta))$, we first aim to find 
$((x_{\alpha}^{(1)}, x_{\beta}^{(2)}),( y_{\gamma}^{(1)}, y_{\delta}^{(2)}))$, or equivalently, 
the optimal $((k_g^{(1)}, k_g^{(2)}) , (k_b^{(1)} , k_b^{(2)}))$ corresponding to such a profile.
%
%
So by incorporating $((\alpha,\beta),(\gamma,\delta))$, Expression (\ref{eqn:raw2phase2camp}) for $\sum_i v_i^{(2)}$ can be simplified as

\begin{small}
\begin{align}
\nonumber
\sum_i \sum_j & c_i b_{ji} 
+ k_g^{(2)} \frac{\theta_\beta}{2} \Big( \sum_i  c_i b_{\beta i} + r_\beta\Big)
+ k_b^{(2)} \frac{\theta_\delta}{2} \Big( \sum_i c_i b_{\delta i} -r_\delta\Big)
\\&
\nonumber
+ k_g^{(1)}  \frac{\theta_\alpha}{2} (1+c_\alpha) \Big( s_{\alpha} 
+ k_g^{(2)} \frac{\theta_\beta}{2} b_{\beta \alpha}
+ k_b^{(2)} \frac{\theta_\delta}{2} b_{\delta \alpha} \Big)
\\&
- k_b^{(1)} \frac{\theta_\gamma}{2} (1-c_\gamma) \Big( s_\gamma 
+ k_g^{(2)} \frac{\theta_\beta}{2} b_{\beta \gamma}
+ k_b^{(2)} \frac{\theta_\delta}{2} b_{\delta \gamma} \Big)
\label{eqn:twophase_twocamp_new}
\end{align}
\end{small}

First, we consider the case when $k_g^{(1)} + k_g^{(2)} = k_g$ and $k_b^{(1)} + k_b^{(2)} = k_b$.
Now, for a given profile of nodes $((\alpha,\beta),(\gamma,\delta))$, we will find the optimal values of $k_g^{(1)},k_g^{(2)},k_b^{(1)},k_b^{(2)}$. In this case, we have $k_g^{(2)} = k_g - k_g^{(1)}$ and $k_b^{(2)} = k_b - k_b^{(1)}$. 
Substituting this in  (\ref{eqn:twophase_twocamp_new}) and simultaneously solving 
$\frac{\partial \sum_i v_i^{(2)} }{\partial k_g^{(1)}}=0$ and
$\frac{\partial \sum_i v_i^{(2)} }{\partial k_b^{(1)}}=0$, we get
$k_g^{(1)}$ equal to
(letting $A = \theta_\gamma \theta_\delta (1-c_\gamma) b_{\delta \gamma}$ and $B = \frac{1}{2} \left( \theta_\alpha \theta_\delta (1+c_\alpha)b_{\delta \alpha} - \theta_\gamma \theta_\beta (1-c_\gamma)b_{\beta \gamma} \right)$ ):

\begin{footnotesize}
\begin{align*}
\hspace{-3mm}
\nonumber
\frac{1}{ B^2 \! + \! \theta_\alpha \theta_\beta  (1 \! + \! c_\alpha)b_{\beta \alpha} A}
\Big[
s_\alpha \theta_\alpha (1 \! + \! c_\alpha)A 
\! - \! r_\beta \theta_\beta A
\! - \! s_\gamma \theta_\gamma (1 \! - \! c_\gamma)B
 \! + \! r_\delta \theta_\delta B
\\
\nonumber
+ k_g\Big(\frac{\theta_\alpha\theta_\beta}{2} (1+c_\alpha)b_{\beta\alpha}A - \frac{\theta_\gamma\theta_\beta}{2} (1-c_\gamma)b_{\beta\gamma}B\Big)
\\
\nonumber
+ k_b \Big(\frac{\theta_\alpha\theta_\delta}{2} (1+c_\alpha)b_{\delta\alpha}A - \frac{\theta_\gamma\theta_\delta}{2} (1-c_\gamma)b_{\delta\gamma}B\Big)
\\
-\theta_\beta A \sum_i c_i b_{\beta i} - \theta_\delta B \sum_i c_i b_{\delta i}
\Big]
\end{align*}
\end{footnotesize}

\begin{figure*}
\centering
\begin{tabular}{ccc}
\hspace{-3mm}
\includegraphics[width=0.33\textwidth]{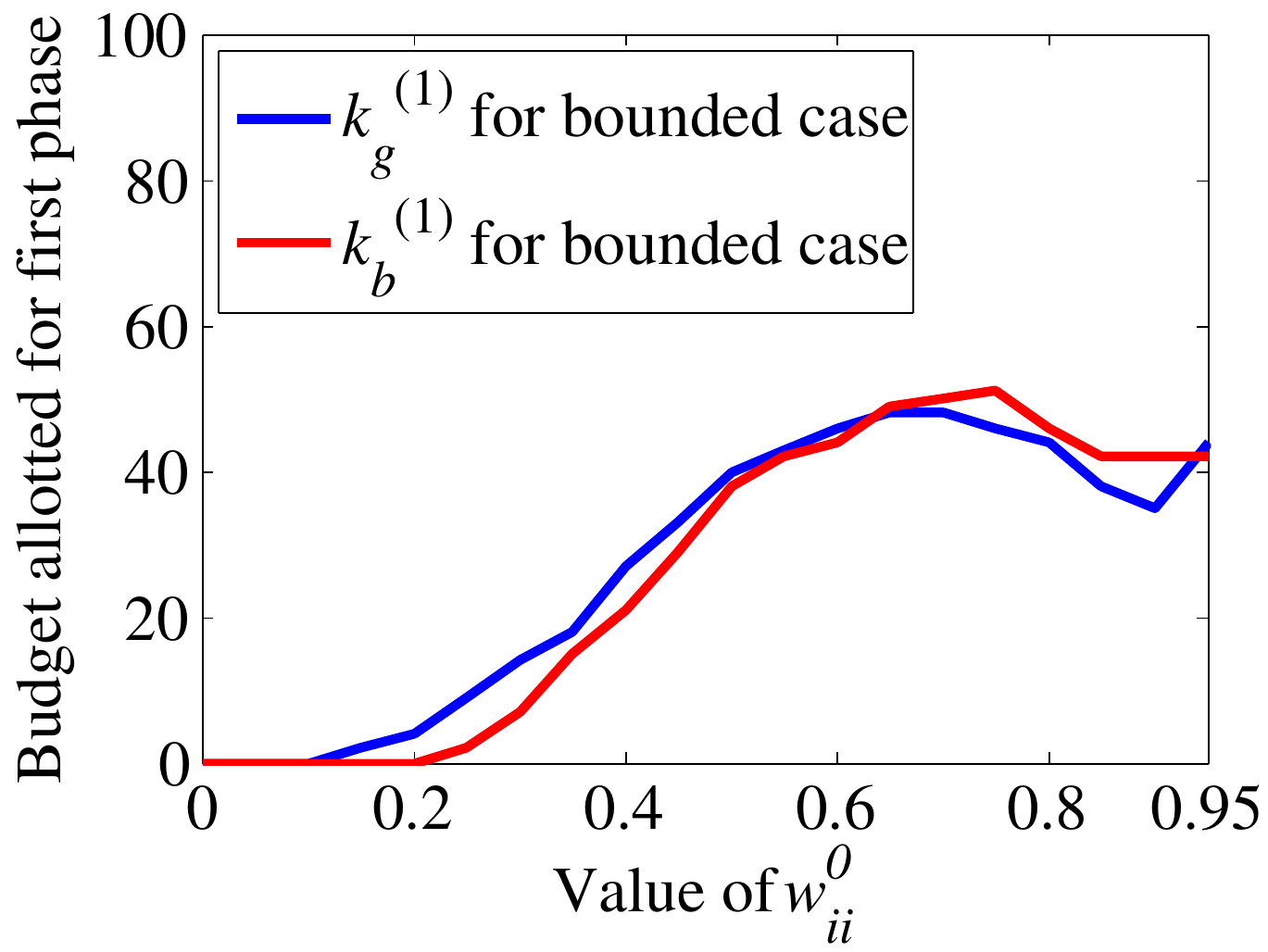}
&
\hspace{-3mm}
\includegraphics[width=0.33\textwidth]{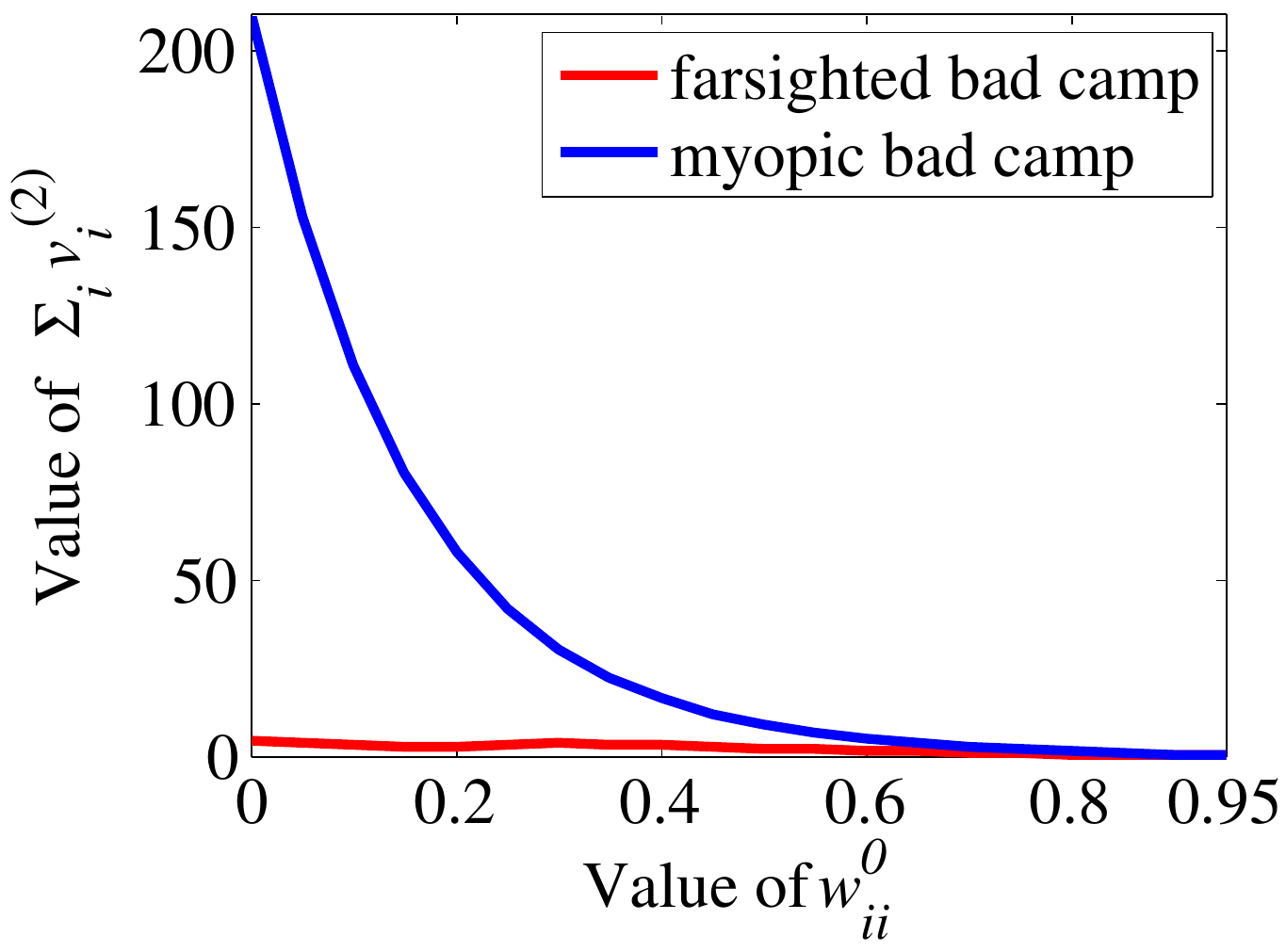}
&
\hspace{-3mm}
\includegraphics[width=0.33\textwidth]{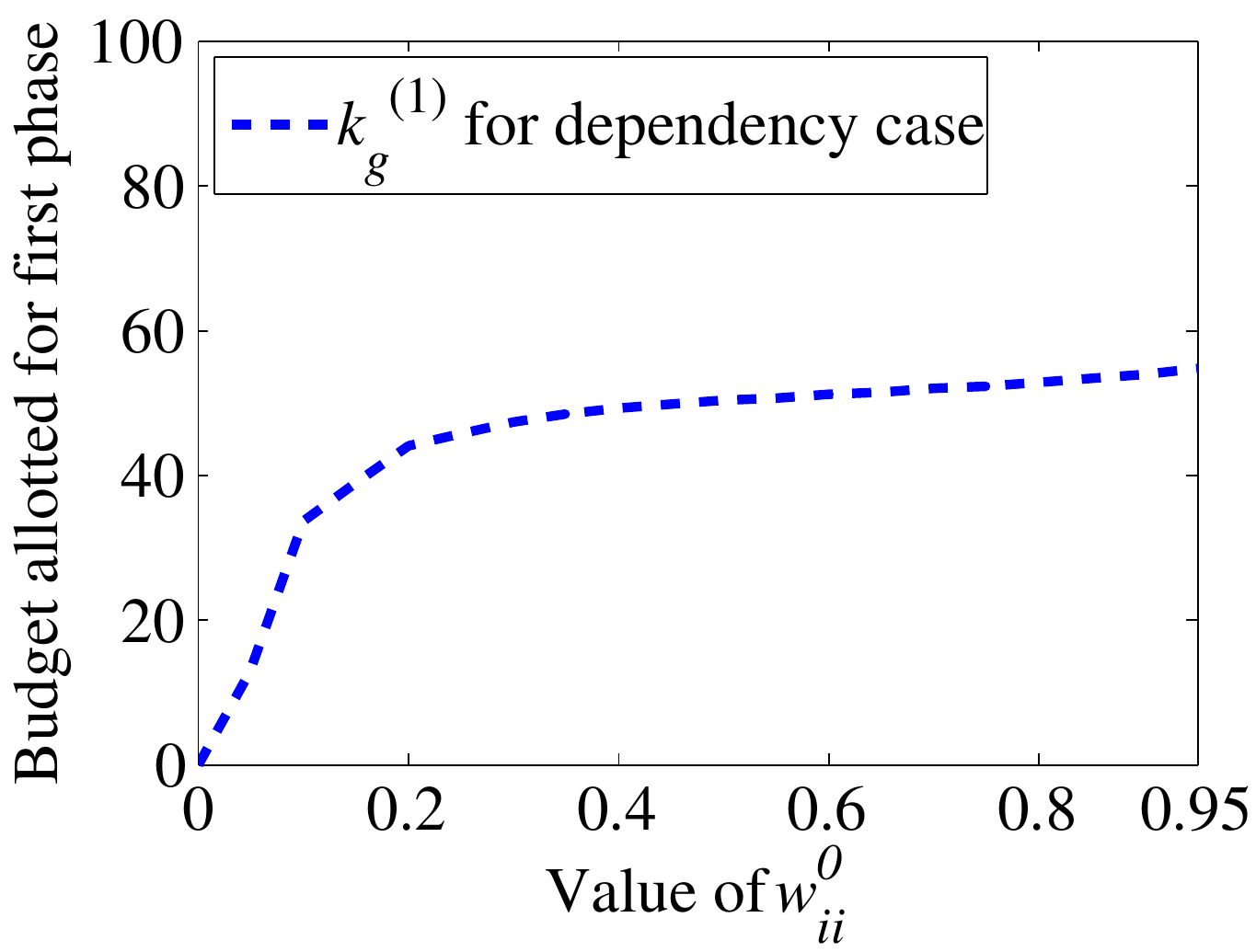}
\\
\hspace{-3mm}
(a) Optimal budget split in 
&
\hspace{-3mm}
(b) Myopic vs farsighted bad camp in
&
\hspace{-3mm}
(c) Optimal budget split in
\\
bounded setting ($k_g=k_b=100$)
&
 bounded setting ($k_g=k_b=100$)
&
 dependency setting ($k_g = 100$)
\end{tabular}
\caption{
Results illustrating the effects of
 $w_{ii}^0$ (NetHEPT)}
\label{fig:budgetsplit_hep}
\vspace{-2mm}
\end{figure*}

We can similarly obtain $k_b^{(1)}$.
%
%
%
%
%
If the second derivative with respect to $k_g^{(1)}$, that is,
$- {\theta_\alpha} {\theta_\beta} r_\beta w_{\beta\beta}^0 \Delta_{\beta\alpha}  (1+w_{\alpha \alpha}^0 v_\alpha^0)  < 0$
and that with respect to $k_b^{(1)}$, that is,
${\theta_\gamma} {\theta_\delta} r_\delta w_{\delta\delta}^0 \Delta_{\delta\gamma}  (1-w_{\gamma \gamma}^0 v_\gamma^0) > 0$, 
and the obtained solution is such that $k_g^{(1)} \in [0,k_g]$ and $k_b^{(1)} \in [0,k_b]$, then we have 
that neither the good camp can change $k_g^{(1)}$ to increase $\sum_i v_i^{(2)}$, nor the bad camp can change $k_b^{(1)}$ to decrease $\sum_i v_i^{(2)}$.
%
%
%
So we can effectively write $u_g((\mathbf{x^{(1)}},\mathbf{x^{(2)}}),(\mathbf{y^{(1)}},\mathbf{y^{(2)}}))$
as $u_g((\alpha,\beta),(\gamma,\delta))$, where $u_g((\alpha,\beta),(\gamma,\delta))$ is the value of $\sum_i v_i^{(2)}$,
which corresponds to the strategy profile where the good camp invests on nodes $(\alpha,\beta)$ with  optimal budget split $(k_g^{(1)},k_g^{(2)})$, and the bad camp invests on nodes $(\gamma,\delta)$ with  optimal budget split $(k_b^{(1)},k_b^{(2)})$.


For a general case where the solution $k_g^{(1)},k_b^{(1)}$ obtained above may not satisfy $k_g^{(1)} \in [0,k_g]$ and $k_b^{(1)} \in [0,k_b]$, we make a practically reasonable assumption so as to determine $u_g((\alpha,\beta),(\gamma,\delta))$.
It is easy to show that if $w_{ij} \geq 0, \forall (i,j)$, we would have $\Delta_{ij} \geq 0, \forall (i,j)$ and $r_i \geq 1, \forall i$.
So if we assume that $w_{ij} \geq 0, \forall (i,j)$ and $w_{ii}^0 \geq 0,\theta_i \geq 0, v_i^0 \in [-1,1], \forall i$,
we would have that 
$- {\theta_\alpha} {\theta_\beta} r_\beta w_{\beta\beta}^0 \Delta_{\beta\alpha}  (1+w_{\alpha \alpha}^0 v_\alpha^0)  \leq 0$
and ${\theta_\gamma} {\theta_\delta} r_\delta w_{\delta\delta}^0 \Delta_{\delta\gamma}  (1-w_{\gamma \gamma}^0 v_\gamma^0) \geq 0$.
That is, we would have $\sum_i v_i^{(2)}$ to be a convex-concave function, which is concave with respect to $k_g^{(1)}$ and convex with respect to $k_b^{(1)}$.
So in the domain $([0,k_g],[0,k_b])$, we can find a $(k_g^{(1)},k_b^{(1)})$ such that,
neither the good camp can change $k_g^{(1)}$ to increase $\sum_i v_i^{(2)}$, nor the bad camp can change $k_b^{(1)}$ to decrease $\sum_i v_i^{(2)}$ \cite{boyd2004convex,arrow1958studies}.
So we can assign this value $\sum_i v_i^{(2)}$ to $u_g((\alpha,\beta),(\gamma,\delta))$.

Thus using the above technique, we obtain  $u_g((\alpha,\beta),(\gamma,\delta))$ for all profiles of nodes $((\alpha,\beta),(\gamma,\delta))$ in
the case when $k_g^{(1)} + k_g^{(2)} = k_g$ and $k_b^{(1)} + k_b^{(2)} = k_b$.
From Lemma \ref{lem:allornothing2camps}, the only other cases to be considered are $k_b^{(1)} = k_b^{(2)} = 0$ and $k_g^{(1)} = k_g^{(2)} = 0$.
Let the profile of nodes $((\alpha,\beta),(0,0))$ correspond to $k_b^{(1)} = k_b^{(2)} = 0$.
Note that when  $k_b^{(1)} = k_b^{(2)} = 0$, it reduces to the single camp case with only the good camp (Section \ref{sec:dep_onecamp}). The value of $\sum_i v_i^{(2)}$ for an $(\alpha,\beta)$ pair can hence be assigned to $u_g((\alpha,\beta),(0,0))$.
Thus we can obtain $u_g((\alpha,\beta),(\gamma,\delta))$ for all profiles of nodes $((\alpha,\beta),(0,0))$.
Similarly, we can obtain $u_g((\alpha,\beta),(\gamma,\delta))$ for all profiles of nodes $((0,0),(\gamma,\delta))$.
And from Equation (\ref{eqn:twophase_twocamp_new}), we know that $u_g((0,0),(0,0)) = \sum_i \sum_j c_i b_{ji} $.

So we have that the good camp has $(n^2+1)$ possible pure strategies to choose from, namely, $(\alpha,\beta) \in N \times N \cup \{(0,0)\}$.
Similarly, the bad camp has $(n^2+1)$ possible pure strategies to choose from, namely, $(\gamma,\delta) \in N \times N \cup \{(0,0)\}$.
We thus have a two-player zero-sum game, for which the utilities of the players can be computed for each strategy profile $((\alpha,\beta),(\gamma,\delta))$ as explained above.
Though we cannot ensure the existence of a pure strategy Nash equilibrium, the finiteness of the number of strategies ensures the existence of a mixed strategy Nash equilibrium. Furthermore, owing to it being a two-player zero-sum game, the Nash equilibria can be found efficiently using linear programming \cite{osborne2004introduction}.

Summarizing, 
under practically reasonable assumptions
($w_{ij} \geq 0, \forall (i,j)$ and $w_{ii}^0 \geq 0,\theta_i \geq 0, v_i^0 \in [-1,1], \forall i$),
we transformed the problem into a two-player zero-sum game with each player having $(n^2+1)$ pure strategies,
and showed how the players' utilities can be computed for each strategy profile.
We
thus deduced the existence of Nash equilibria and that they can be found efficiently
using linear programming.

\vspace{-2mm}
\section{Simulations and Results}
\label{sec:ODSNmultiphase_sim}

%
For determining implications of our analytically derived results on real-world networks,
we conducted a simulation study on 
a popular network dataset:
NetHEPT (a co-authorship network in the ``High Energy Physics - Theory'' papers on the e-print arXiv from 1991 to 2003),
consisting of 15,233 nodes and 31,376 edges. This network exhibits many structural features of large-scale social networks  and is widely used for experimental justifications, for example, in 
\cite{kempe2003maximizing,chen2009efficient,chen2010scalable}.
In our simulations,
we consider $v_i^0=0, \forall i$.
We transform the original undirected network into a directed one by making all edges bidirectional.
Let $N(i)$ be the set of neighbors of node $i$ in the network.
For a fixed $w_{ii}^0$, we consider $w_{ig},w_{ib}>0$ and $w_{ij} \geq 0, \forall j \in N(i)$ such that
the weights for any node $i$ sum to 1, that is, $w_{ig}+w_{ib}+\sum_{j\in N(i)} w_{ij}=1-w_{ii}^0$.
Furthermore, for different values of $w_{ii}^0$, the values of $w_{ig},w_{ib},w_{ij}$ are scaled proportional to $1-w_{ii}^0$.
That is, the values $\hat{w}_{ig},\hat{w}_{ib},\hat{w}_{ij}$ corresponding to $\hat{w}_{ii}^0$ and the values $\tilde{w}_{ig},\tilde{w}_{ib},\tilde{w}_{ij}$ corresponding to $\tilde{w}_{ii}^0$, are such that

\begin{small}
\vspace{-3mm}
\begin{align*}
\forall i:\;
\frac{1-\hat{w}_{ii}^0}{1-\tilde w_{ii}^0}  = \frac{\hat{w}_{ig}}{\tilde w_{ig}} = \frac{\hat{w}_{ib}}{\tilde w_{ib}} = \frac{\hat{w}_{ij}}{\tilde w_{ij}}, \; \forall j \in N(i)
\end{align*}
\end{small}
%
%
%
%
%
%
%
%
For the case of dependency among parameters, where $w_{ig}^{(q)} = \frac{\theta_i}{2} ( {1+w_{ii}^0 v_i^{(q-1)}} )$ and 
$w_{ib}^{(q)} = \frac{\theta_i}{2} ( {1-w_{ii}^0 v_i^{(q-1)}} )$,
we let 
$\theta_i = \tilde{w}_{ig}+\tilde{w}_{ib}$,
where $\tilde{w}_{ig}$ and $\tilde{w}_{ib}$ are as above.
Furthermore, since we assume $v_i^0=0$, we would have $w_{ig}^{(1)} = w_{ib}^{(1)} = \frac{\theta_i}{2}$.
We consider different values of $w_{ii}^0$, namely, $\{0,0.05,\ldots,0.95\}$.
In our simulation study, we assume that this value is same for all nodes.
This allows us to study the effect of the values of $w_{ii}^0$ on the two phase strategy of the camps.





\textbf{Simulation Results.}
Figures \ref{fig:budgetsplit_hep} (a) and \ref{fig:budgetsplit_hep}(c)
present the optimal amount of budget that should be invested in the first phase as a function of $w_{ii}^0$ for the bounded setting (where investment on each node is bounded in either phase) and the dependency setting (where the influence weight of the camp on a node depends on its initial bias), respectively. In our simulations, the optimal values obtained are such that $k_g^{(2)}=k_g-k_g^{(1)}$ and $k_b^{(2)}=k_b-k_b^{(1)}$.
For low values of $w_{ii}^0$,
the optimal strategy of the camps is to invest almost entirely in the second phase. This is because the effect of first phase diminishes in second phase. 
The value $s_j = \sum_i r_i w_{ii}^0 \Delta_{ij}$ of a node $j$ would be significant only if it influences nodes $j$ with significant values of $w_{ii}^0$. And hence, investing in first phase would be advantageous only if we have nodes which have significant value of $w_{ii}^0$.
With low values to $w_{ii}^0$; we are less likely to have nodes $j$ with high values of $s_j$ since it requires them to not only be influential but also influential towards significant number of nodes $i$ with significant values of $w_{ii}^0$.
%
This also justifies Figure \ref{fig:budgetsplit_hep}(b) where, if the bad camp acts myopically (spending its entire budget in the first phase), it incurs significantly more loss for lower values of $w_{ii}^0$ (higher value of $\sum_i v_i^{(2)}$ is worse for the bad camp).

Though we have observed that low values of $w_{ii}^0$ are detrimental for the values of $s_j$ and hence for investment in the first phase, the budget allotted to the first phase need not be a monotone function of $w_{ii}^0$.
It is possible that  for a value of $w_{ii}^0$, a set of nodes have good values of $s_j=\sum_i r_i w_{ii}^0 \Delta_{ij}$ and hence are worth investing on in the first phase. However for a higher range of $w_{ii}^0$, the values of $w_{ig},w_{ib},w_{ij}$ scale down, and hence the first phase worth of a node $j$ ($s_j w_{jg}$ or $s_j w_{jb}$) may be superseded by the second phase worth of some other node $j'$ ($r_{j'}w_{j'g}$ or $r_{j'}w_{j'b}$), thus possibly making it advantageous to invest in the second phase for a higher $w_{ii}^0$.
%
%
In Figure \ref{fig:budgetsplit_hep}(a), for the setting where the values of $w_{ig}$ and $w_{ib}$ do not depend on 
the initial opinion  of node $i$ (when we consider bounded investment per node),
the general observations suggest that 
a high range of $w_{ii}^0$ makes it advantageous for the camps to invest in the first phase, albeit in a non-monotone way owing to the aforementioned reason.
In Figure \ref{fig:budgetsplit_hep}(c), for the dependency setting  with single camp (where we consider unbounded investment per node),
a high range of $w_{ii}^0$ makes it advantageous for the camps to invest in the first phase.
Here we point  that though we presented a polynomial time algorithm for the dependency setting with two camps, which is of theoretical interest, it is computationally expensive to run on networks larger than a few hundred nodes.

%


%

\vspace{-2mm}
\section{Conclusion
}

Using DeGroot-Friedkin model of opinion dynamics, we proposed a framework for optimal multiphase investment strategies for two competing camps in a social network. 
%
We focused on two phases and derived closed-form expressions for optimal strategies, and 
 the extent of loss that a camp would incur if it acted myopically 
against a farsighted competitor.
We then studied the setting where the influence of the camps on a node depended on the node's initial bias. For the case of single camp, we derived polynomial time algorithm for determining an optimal way of splitting the  budget between the two phases. 
%
%
For the case of two camps, we showed the existence of Nash equilibria under reasonable assumptions, and that they can be computed in polynomial time.
Our simulations quantified the impact of the weightage that nodes attribute to their initial biases. Notably, higher weightage would necessitate more investment in the first phase, so as to influence these biases for the second (terminal) phase.

As future direction,
 it would be interesting to study optimal multiphase investment strategies under other models of opinion dynamics in the literature.
%
The two camps setting can be extended to multiple camps.
%
Also, the two phase study can be generalized to multiple phases to see if any important insights or benefits can be obtained.
%
%
It would be interesting to study the problem with bounds on total investment on a node by the two camps together (such as $\forall i,\, x_i+y_i \leq 1$).
It is also worth exploring the possibility of more efficient algorithms for the dependency setting with two camps.

\vspace{-2mm}
\bibliographystyle{IEEEtran}
\bibliography{ODSN_Multiphase_references} 

\begin{thebibliography}{10}
\providecommand{\url}[1]{#1}
\csname url@samestyle\endcsname
\providecommand{\newblock}{\relax}
\providecommand{\bibinfo}[2]{#2}
\providecommand{\BIBentrySTDinterwordspacing}{\spaceskip=0pt\relax}
\providecommand{\BIBentryALTinterwordstretchfactor}{4}
\providecommand{\BIBentryALTinterwordspacing}{\spaceskip=\fontdimen2\font plus
\BIBentryALTinterwordstretchfactor\fontdimen3\font minus
  \fontdimen4\font\relax}
\providecommand{\BIBforeignlanguage}[2]{{%
\expandafter\ifx\csname l@#1\endcsname\relax
\typeout{** WARNING: IEEEtran.bst: No hyphenation pattern has been}%
\typeout{** loaded for the language `#1'. Using the pattern for}%
\typeout{** the default language instead.}%
\else
\language=\csname l@#1\endcsname
\fi
#2}}
\providecommand{\BIBdecl}{\relax}
\BIBdecl

\bibitem{gionis2013opinion}
A.~Gionis, E.~Terzi, and P.~Tsaparas, ``Opinion maximization in social
  networks,'' in \emph{Proceedings of the 2013 International Conference on Data
  Mining}.\hskip 1em plus 0.5em minus 0.4em\relax SIAM, 2013, pp. 387--395.

\bibitem{grabisch2017strategic}
M.~Grabisch, A.~Mandel, A.~Rusinowska, and E.~Tanimura, ``Strategic influence
  in social networks,'' \emph{Mathematics of Operations Research}, 2017.

\bibitem{easley2010networks}
D.~Easley and J.~Kleinberg, ``Networks, crowds, and markets,'' \emph{Cambridge
  Univ Press}, vol.~6, no.~1, pp. 6--1, 2010.

\bibitem{acemoglu2011opinion}
D.~Acemoglu and A.~Ozdaglar, ``Opinion dynamics and learning in social
  networks,'' \emph{Dynamic Games and Applications}, vol.~1, no.~1, pp. 3--49,
  2011.

\bibitem{friedkin1990social}
N.~E. Friedkin and E.~C. Johnsen, ``Social influence and opinions,''
  \emph{Journal of Mathematical Sociology}, vol.~15, no. 3-4, pp. 193--206,
  1990.

\bibitem{friedkin1997social}
------, ``Social positions in influence networks,'' \emph{Social Networks},
  vol.~19, no.~3, pp. 209--222, 1997.

\bibitem{chatterjee2013predicting}
S.~Chatterjee, F.~M. Hafizoglu, and S.~Sen, ``Predicting migration and opinion
  adoption patterns in agent communities,'' in \emph{Proceedings of the 2013
  International Conference on Autonomous Agents \& Multiagent Systems}.\hskip
  1em plus 0.5em minus 0.4em\relax IFAAMAS, 2013, pp. 1287--1288.

\bibitem{crawford2013opposites}
C.~Crawford, L.~Brooks, and S.~Sen, ``Opposites repel: the effect of
  incorporating repulsion on opinion dynamics in the bounded confidence
  model,'' in \emph{Proceedings of the 2013 International Conference on
  Autonomous Agents \& Multiagent Systems}.\hskip 1em plus 0.5em minus
  0.4em\relax IFAAMAS, 2013, pp. 1225--1226.

\bibitem{grandi2017strategic}
U.~Grandi, E.~Lorini, A.~Novaro, and L.~Perrussel, ``Strategic disclosure of
  opinions on a social network,'' in \emph{Proceedings of the 16th Conference
  on Autonomous Agents \& Multiagent Systems}.\hskip 1em plus 0.5em minus
  0.4em\relax IFAAMAS, 2017, pp. 1196--1204.

\bibitem{grandi2015propositional}
U.~Grandi, E.~Lorini, and L.~Perrussel, ``Propositional opinion diffusion,'' in
  \emph{Proceedings of the 2015 International Conference on Autonomous Agents
  \& Multiagent Systems}.\hskip 1em plus 0.5em minus 0.4em\relax IFAAMAS, 2015,
  pp. 989--997.

\bibitem{sina2015adapting}
S.~Sina, N.~Hazon, A.~Hassidim, and S.~Kraus, ``Adapting the social network to
  affect elections,'' in \emph{Proceedings of the 2015 International Conference
  on Autonomous Agents \& Multiagent Systems}.\hskip 1em plus 0.5em minus
  0.4em\relax IFAAMAS, 2015, pp. 705--713.

\bibitem{soriano2016simultaneous}
L.~Soriano~Marcolino, A.~Lakshminarayanan, A.~Yadav, and M.~Tambe,
  ``Simultaneous influencing and mapping social networks,'' in
  \emph{Proceedings of the 2016 International Conference on Autonomous Agents
  \& Multiagent Systems}.\hskip 1em plus 0.5em minus 0.4em\relax IFAAMAS, 2016,
  pp. 1439--1440.

\bibitem{tsang2014opinion}
A.~Tsang and K.~Larson, ``Opinion dynamics of skeptical agents,'' in
  \emph{Proceedings of the 2014 International Conference on Autonomous Agents
  \& Multiagent Systems}.\hskip 1em plus 0.5em minus 0.4em\relax IFAAMAS, 2014,
  pp. 277--284.

\bibitem{yadav2016using}
A.~Yadav, H.~Chan, A.~Xin~Jiang, H.~Xu, E.~Rice, and M.~Tambe, ``Using social
  networks to aid homeless shelters: Dynamic influence maximization under
  uncertainty,'' in \emph{Proceedings of the 2016 International Conference on
  Autonomous Agents \& Multiagent Systems}.\hskip 1em plus 0.5em minus
  0.4em\relax IFAAMAS, 2016, pp. 740--748.

\bibitem{lorenz2007continuous}
J.~Lorenz, ``Continuous opinion dynamics under bounded confidence: A survey,''
  \emph{International Journal of Modern Physics C}, vol.~18, no.~12, pp.
  1819--1838, 2007.

\bibitem{degroot1974reaching}
M.~H. DeGroot, ``Reaching a consensus,'' \emph{Journal of the American
  Statistical Association}, vol.~69, no. 345, pp. 118--121, 1974.

\bibitem{clifford1973model}
P.~Clifford and A.~Sudbury, ``A model for spatial conflict,''
  \emph{Biometrika}, vol.~60, no.~3, pp. 581--588, 1973.

\bibitem{holley1975ergodic}
R.~A. Holley and T.~M. Liggett, ``Ergodic theorems for weakly interacting
  infinite systems and the voter model,'' \emph{The annals of probability}, pp.
  643--663, 1975.

\bibitem{krause2000discrete}
U.~Krause, ``A discrete nonlinear and non-autonomous model of consensus
  formation,'' \emph{Communications in difference equations}, pp. 227--236,
  2000.

\bibitem{guille2013information}
A.~Guille, H.~Hacid, C.~Favre, and D.~A. Zighed, ``Information diffusion in
  online social networks: A survey,'' \emph{ACM SIGMOD Record}, vol.~42, no.~1,
  pp. 17--28, 2013.

\bibitem{kempe2003maximizing}
D.~Kempe, J.~Kleinberg, and {\'E}.~Tardos, ``Maximizing the spread of influence
  through a social network,'' in \emph{Proceedings of the ninth ACM SIGKDD
  International Conference on Knowledge Discovery and Data Mining}.\hskip 1em
  plus 0.5em minus 0.4em\relax ACM, 2003, pp. 137--146.

\bibitem{amoruso2017contrasting}
M.~Amoruso, D.~Anello, V.~Auletta, and D.~Ferraioli, ``Contrasting the spread
  of misinformation in online social networks,'' in \emph{The 16th Conference
  on Autonomous Agents \& Multiagent Systems}.\hskip 1em plus 0.5em minus
  0.4em\relax IFAAMAS, 2017, pp. 1323--1331.

\bibitem{cholvy2016influence}
L.~Cholvy, ``Influence-based opinion diffusion,'' in \emph{Proceedings of the
  2016 International Conference on Autonomous Agents \& Multiagent
  Systems}.\hskip 1em plus 0.5em minus 0.4em\relax IFAAMAS, 2016, pp.
  1355--1356.

\bibitem{ghanem2012agents}
A.~G. Ghanem, S.~Vedanarayanan, and A.~A. Minai, ``Agents of influence in
  social networks,'' in \emph{Proceedings of the 11th International Conference
  on Autonomous Agents \& Multiagent Systems-Volume 1}.\hskip 1em plus 0.5em
  minus 0.4em\relax IFAAMAS, 2012, pp. 551--558.

\bibitem{li2017agent}
W.~Li, Q.~Bai, T.~D. Nguyen, and M.~Zhang, ``Agent-based influence maintenance
  in social networks,'' in \emph{Proceedings of the 16th Conference on
  Autonomous Agents \& Multiagent Systems}.\hskip 1em plus 0.5em minus
  0.4em\relax IFAAMAS, 2017, pp. 1592--1594.

\bibitem{maghami2012identifying}
M.~Maghami and G.~Sukthankar, ``Identifying influential agents for advertising
  in multi-agent markets,'' in \emph{Proceedings of the 11th International
  Conference on Autonomous Agents \& Multiagent Systems-Volume 2}.\hskip 1em
  plus 0.5em minus 0.4em\relax IFAAMAS, 2012, pp. 687--694.

\bibitem{pasumarthi2015near}
R.~Pasumarthi, R.~Narayanam, and B.~Ravindran, ``Near optimal strategies for
  targeted marketing in social networks,'' in \emph{Proceedings of the 2015
  International Conference on Autonomous Agents \& Multiagent Systems}.\hskip
  1em plus 0.5em minus 0.4em\relax IFAAMAS, 2015, pp. 1679--1680.

\bibitem{yildiz2013binary}
E.~Yildiz, A.~Ozdaglar, D.~Acemoglu, A.~Saberi, and A.~Scaglione, ``Binary
  opinion dynamics with stubborn agents,'' \emph{ACM Transactions on Economics
  and Computation}, vol.~1, no.~4, p.~19, 2013.

\bibitem{ghaderi2014opinion}
J.~Ghaderi and R.~Srikant, ``Opinion dynamics in social networks with stubborn
  agents: Equilibrium and convergence rate,'' \emph{Automatica}, vol.~50,
  no.~12, pp. 3209--3215, 2014.

\bibitem{anagnostopoulos2015competitive}
A.~Anagnostopoulos, D.~Ferraioli, and S.~Leonardi, ``Competitive influence in
  social networks: Convergence, submodularity, and competition effects,'' in
  \emph{Proceedings of the 2015 International Conference on Autonomous Agents
  \& Multiagent Systems}.\hskip 1em plus 0.5em minus 0.4em\relax IFAAMAS, 2015,
  pp. 1767--1768.

\bibitem{bharathi2007competitive}
S.~Bharathi, D.~Kempe, and M.~Salek, ``Competitive influence maximization in
  social networks,'' in \emph{International Workshop on Web and Internet
  Economics}.\hskip 1em plus 0.5em minus 0.4em\relax Springer, 2007, pp.
  306--311.

\bibitem{goyal2014competitive}
S.~Goyal, H.~Heidari, and M.~Kearns, ``Competitive contagion in networks,''
  \emph{Games and Economic Behavior}, 2014.

\bibitem{dubey2006competing}
P.~Dubey, R.~Garg, and B.~De~Meyer, ``Competing for customers in a social
  network: The quasi-linear case,'' in \emph{WINE}.\hskip 1em plus 0.5em minus
  0.4em\relax Springer, 2006, pp. 162--173.

\bibitem{bimpikis2016competitive}
K.~Bimpikis, A.~Ozdaglar, and E.~Yildiz, ``Competitive targeted advertising
  over networks,'' \emph{Operations Research}, vol.~64, no.~3, pp. 705--720,
  2016.

\bibitem{singer2016influence}
Y.~Singer, ``Influence maximization through adaptive seeding,'' \emph{ACM
  SIGecom Exchanges}, vol.~15, no.~1, pp. 32--59, 2016.

\bibitem{horel2015scalable}
T.~Horel and Y.~Singer, ``Scalable methods for adaptively seeding a social
  network,'' in \emph{Proceedings of the 24th International Conference on World
  Wide Web}.\hskip 1em plus 0.5em minus 0.4em\relax ACM, 2015, pp. 441--451.

\bibitem{dhamal2015multiphase}
S.~Dhamal, P.~K~J, and Y.~Narahari, ``A multi-phase approach for improving
  information diffusion in social networks,'' in \emph{Proceedings of the 14th
  International Conference on Autonomous Agents \& Multiagent Systems}.\hskip
  1em plus 0.5em minus 0.4em\relax IFAAMAS, 2015, pp. 1787--1788.

\bibitem{dhamal2016information}
S.~Dhamal, K.~J. Prabuchandran, and Y.~Narahari, ``Information diffusion in
  social networks in two phases,'' \emph{IEEE Transactions on Network Science
  and Engineering}, vol.~3, no.~4, pp. 197--210, 2016.

\bibitem{altafini2013consensus}
C.~Altafini, ``Consensus problems on networks with antagonistic interactions,''
  \emph{IEEE Transactions on Automatic Control}, vol.~58, no.~4, pp. 935--946,
  2013.

\bibitem{boyd2004convex}
S.~Boyd and L.~Vandenberghe, \emph{Convex optimization}.\hskip 1em plus 0.5em
  minus 0.4em\relax Cambridge university press, 2004.

\bibitem{arrow1958studies}
K.~J. Arrow, L.~Hurwicz, H.~Uzawa, and H.~B. Chenery, ``Studies in linear and
  non-linear programming,'' 1958.

\bibitem{osborne2004introduction}
M.~J. Osborne, \emph{An introduction to game theory}.\hskip 1em plus 0.5em
  minus 0.4em\relax Oxford university press New York, 2004, vol.~3, no.~3.

\bibitem{chen2009efficient}
W.~Chen, Y.~Wang, and S.~Yang, ``Efficient influence maximization in social
  networks,'' in \emph{Proceedings of the 15th ACM SIGKDD International
  Conference on Knowledge Discovery and Data Mining}.\hskip 1em plus 0.5em
  minus 0.4em\relax ACM, 2009, pp. 199--208.

\bibitem{chen2010scalable}
W.~Chen, C.~Wang, and Y.~Wang, ``Scalable influence maximization for prevalent
  viral marketing in large-scale social networks,'' in \emph{Proceedings of the
  16th ACM SIGKDD International Conference on Knowledge Discovery and Data
  Mining}.\hskip 1em plus 0.5em minus 0.4em\relax ACM, 2010, pp. 1029--1038.

\end{thebibliography}



\end{document}